\documentclass{article}

\PassOptionsToPackage{numbers,square,comma}{natbib}
\usepackage[main, final]{neurips_2026}

\usepackage[utf8]{inputenc} 
\usepackage[T1]{fontenc}    
\usepackage{hyperref}       
\usepackage{url}            
\usepackage{booktabs}       
\usepackage{amsfonts}       
\usepackage{nicefrac}       
\usepackage{microtype}      
\usepackage{xcolor}         

\newcommand{\ignore}[1]{}

\usepackage{amsmath}
\usepackage{mathtools}
\usepackage{amssymb}
\usepackage{caption}
\usepackage{algorithmic}

\usepackage[ruled,vlined]{algorithm2e}
\usepackage{graphicx}
\usepackage{textcomp}
\usepackage{xcolor}
\usepackage{verbatim}
\usepackage{subcaption}
\usepackage{adjustbox}
\usepackage{booktabs}
\usepackage{multirow}
\usepackage{threeparttable, tablefootnote}
\usepackage{varwidth,xcolor}

\usepackage{tabularray}
\usepackage{ragged2e}
\usepackage{array}
\usepackage{bm}
\usepackage{mathrsfs}
\usepackage{tabularx}
\usepackage{makecell}
\usepackage{diagbox}

\usepackage{stmaryrd}
\usepackage{amssymb}

\usepackage{amsthm}
\newtheorem{theorem}{Theorem}
\usepackage{thmtools}
\usepackage{thm-restate}
\theoremstyle{remark}

\newtheorem{lemma}[theorem]{Lemma}

\usepackage{float}

\usepackage{multicol}
\usepackage{array}
\usepackage{longtable}
\usepackage{array}
\newcolumntype{L}[1]{>{\raggedright\arraybackslash}p{#1}}

\title{Jaguar: Fast Private Neural Network Inference with Power-of-Two Homomorphic Arithmetic}
\author{%
Yewon Jeong\textsuperscript{*}
\quad
Nayoung Jung\textsuperscript{*}
\quad
Hyeri Roh
\quad
Woo-Seok Choi\textsuperscript{\(\dagger\)}
\\[0.1em]
Department of Electrical and Computer Engineering, ISRC, Seoul National University
\\[-0.0em]
\texttt{\{yellowon518,18170lanige,hrroh,wooseokchoi\}@snu.ac.kr}
\\[-0.0em]
\textsuperscript{*}Equal contribution.
}

\begin{document}

\maketitle

\begingroup
\renewcommand\thefootnote{}
\footnotetext{\textsuperscript{\(\dagger\)}Corresponding author.}
\endgroup

\begin{abstract}
Hybrid HE/2PC private neural network inference remains bottlenecked by prime-modulus homomorphic arithmetic in convolution and by a precision flow that runs ReLU at doubled bitwidth before invoking a separate truncation protocol. We present \textbf{Jaguar}, a system built on a single design choice---a \emph{power-of-two ciphertext ring}---that addresses both. The choice enables \textsc{SPA-Conv}, a coefficient-domain convolution kernel that replaces NTT-centric polynomial multiplication with Scalar--Polynomial Accumulation, and an \emph{exact} ciphertext-side truncation by local right shifts that lets ReLU run directly at the target fixed-point precision and eliminates the post-ReLU truncation protocol. 
Where NTT remains genuinely useful---at the client, for the single polynomial multiplication during decryption---we recover it through an auxiliary NTT prime, preserving the power-of-two protocol substrate while keeping decryption $O(N\log N)$.
On ImageNet-scale ResNet-18, ResNet-50, and MobileNetV2 with AVX disabled, Jaguar achieves 2.07--3.72$\times$ lower end-to-end latency than Cheetah and 2.16--3.36$\times$ lower than Rhombus, with 1.16--1.76$\times$ lower communication than Cheetah. 
Our implementation is publicly available at \url{https://github.com/yellowon/Jaguar}.
\end{abstract}
\section{Introduction}
\label{sec:intro}

Machine learning as a service (MLaaS) is widely deployed for vision~\cite{he2016resnet}, medical~\cite{esteva2017dermatologist,kaissis2021end}, mobile~\cite{sandler2018mobilenetv2}, and embedded~\cite{lin2020mcunet,warden2019tinyml} inference. 
Remote neural network inference creates a two-sided privacy problem---the client wants to protect sensitive inputs, the server wants to protect proprietary model parameters---solved by \emph{private inference} without revealing either side's secrets~\cite{mohassel2017secureml, liu2017minionn, juvekar2018gazelle, mishra2020delphi, rathee2020cryptflow2, huang2022cheetah}.

Two paradigms have emerged.
\emph{Fully homomorphic encryption(HE)-based} inference evaluates the entire network under encryption but pays heavy compute cost and typically replaces nonlinearities with polynomial approximations, degrading accuracy~\cite{giladbachrach2016cryptonets, brutzkus2019lola, lee2022lowcomplexity, ju2024neujeans, stoian2023tfhe, roh2024flash}. 
\emph{Hybrid HE/2PC} inference evaluates linear layers under HE and nonlinear layers under secure two-party computation(2PC)~\cite{juvekar2018gazelle,mishra2020delphi,
rathee2020cryptflow2,huang2022cheetah} (Figure~\ref{fig:intro}(a)). This avoids activation approximation and is the practical design point for exact private neural network inference.

Within this hybrid framework, HE evaluates static-weight linear operators and returns their outputs as additive shares for subsequent secure computation~\cite{huang2022cheetah}.
This linear-to-share interface appears across neural architectures, including convolutions in CNNs and token-wise projections in attention models~\cite{pang2024bolt}. This paper targets the homomorphic arithmetic underlying this interface.
We use CNNs as the running example in the following discussion and as the basis for our end-to-end implementation and evaluation.
This setting allows us to establish the mathematical correctness, noise behavior, and system-level performance of our new arithmetic design within a complete inference pipeline, while leaving the choice of private attention protocols separate. Within this CNN setting, our profiling identifies \emph{two bottlenecks} in the hybrid inference pipeline.

\textbf{The first is HE convolution.}
On ResNet-50 under LAN, our profiling of Cheetah~\cite{huang2022cheetah,opencheetah} shows convolution alone takes 96.71\,s out of 120.82\,s, over 99\,\% of linear-layer latency. 
Within HE convolution, NTT-domain polynomial operations dominate:
forward and inverse NTTs together account for 68\,\% of convolution time, 
and modular prime reductions---appearing both inside and outside NTTs---account for 34.6\,\% (Figure~\ref{fig:intro}(b)).
The bottleneck is not implementation inefficiency---it is the NTT-prime \emph{arithmetic regime} that the protocol was built on.

\textbf{A second, less-discussed bottleneck is precision flow}. A linear layer in fixed-point inference produces output at \emph{doubled} fractional precision ($2f$ bits). Conventional hybrid systems pass this directly into ReLU and then invoke a \emph{separate} truncation protocol~\cite{rathee2020cryptflow2,huang2022cheetah,choi2022impala,gupta2022llama}. This is doubly wasteful: ReLU pays for the wider bitwidth, and an interactive 2PC truncation is added on top. 

\textbf{Why power-of-two.}
Both bottlenecks share a root cause: the \emph{prime-modulus} convention inherited from FHE library design.
NTT-friendly primes are excellent for generic polynomial multiplication, so libraries default to them.
But a hybrid HE/2PC system performs neither generic polynomial multiplication nor deep homomorphic circuits---it performs a single linear layer at a time, immediately converts to additive shares over $\mathbb{Z}_{2^\ell}$, and reuses the share domain for 2PC.
In this regime the prime modulus \emph{blocks} two capabilities the protocol very much wants:
\begin{enumerate}
    \item \textbf{Reduction-free coefficient arithmetic}.
    Convolution can be evaluated as a sum of shifted scalar--polynomial products without ever multiplying two polynomials, \emph{if} coefficient-wise multiplication and addition are cheap.
    Power-of-two arithmetic makes them cheap by replacing modular reduction with a bit mask.
    \item \textbf{Exact ciphertext-side truncation}.
    A right shift on each ciphertext component implements an exact plaintext truncation \emph{iff} the BFV scale $\Delta$ has $2^f$ as a divisor---automatic when both $q$ and $p$ are powers of two, impossible when $q$ is a product of NTT-friendly primes.
\end{enumerate}
This paper proposes \textbf{Jaguar} that adopts $q = 2^Q$, $p = 2^P$ as the central design choice.
\textsc{SPA-Conv} and pre-ReLU ciphertext truncation are then \emph{consequences} of this choice, not independent optimizations.
The remaining technical question is whether giving up NTT in the linear path costs us elsewhere;
it does not, because NTT can be recovered as a local client-side tool exactly where it pays off.

\begin{figure}[t]
    \centering
    \includegraphics[width=\columnwidth]{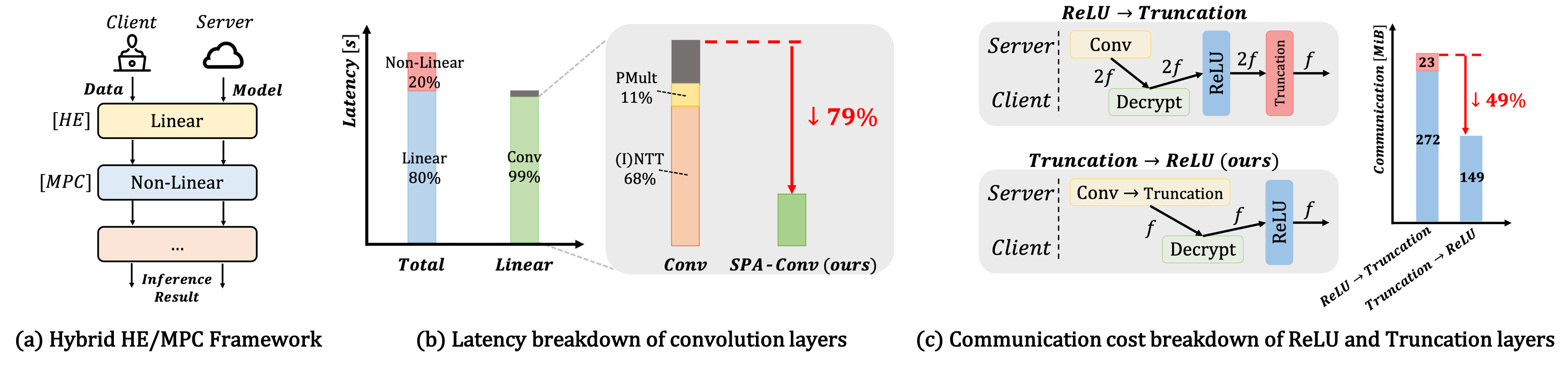}
    \vspace{-1.5em}
    \caption{
    Motivation for Jaguar's design choices.
    (a) Hybrid HE/2PC private inference framework. 
    (b) Latency breakdown of Cheetah's ResNet-50 convolutions under LAN.
    (c) Communication of the ReLU+truncation stage on ResNet-50.
    }
    \label{fig:intro}
\end{figure}

\textbf{Contributions.}
We implement Jaguar on ImageNet-scale ResNet-18, ResNet-50, and MobileNetV2 under both LAN and WAN. With AVX disabled,
end-to-end latency is \textbf{2.07--3.72$\times$ lower than Cheetah} and \textbf{2.16--3.36$\times$ lower than Rhombus}~\cite{he2024Rhombus};
communication is 1.16--1.76$\times$ lower than Cheetah and comparable to Rhombus (1.09$\times$ and 1.41$\times$ lower on ResNet-18 and MobileNetV2; 1.28$\times$ higher on ResNet-50---we discuss this case explicitly in Section~\ref{sec:experiments}).
Convolution-kernel latency drops 79\,\% on ResNet-50 (Figure~\ref{fig:intro}(b));
ReLU+truncation communication drops 1.98$\times$, with the truncation stage entirely eliminated (Figure~\ref{fig:intro}(c)). 
The technical contributions are:

\begin{itemize}
    \item \textbf{A power-of-two HE arithmetic substrate} for hybrid HE/2PC private neural network inference.
    To our knowledge, Jaguar is the first HE/2PC system to use a power-of-two \emph{ciphertext} ring as the central arithmetic regime for ImageNet-scale CNN inference. Prior coefficient-encoded systems used a power-of-two plaintext/share ring but kept the ciphertext backend on NTT-friendly primes.

    \item \textbf{\textsc{SPA-Conv}}, a coefficient-domain convolution kernel that evaluates pointwise, dense, and depthwise convolutions as sums of scalar--polynomial products and coefficient shifts, with no plaintext--ciphertext polynomial multiplication, NTT, or modular prime reduction.

    \item \textbf{Exact pre-ReLU ciphertext-side truncation,} a zero-communication operation that removes the separate post-ReLU truncation protocol used by all prior hybrid HE/2PC systems and reduces the bitwidth at which ReLU is evaluated. We prove its correctness (Theorem~\ref{thm:precise_trunc}).

    \item \textbf{NTT-assisted decryption,} a client-local technique that recovers NTT exactly where it is genuinely useful---the single $\hat{c}'_1\cdot\mathsf{sk}$ multiplication during decryption---via an auxiliary NTT prime, without disturbing the power-of-two protocol substrate.
\end{itemize}
\section{Preliminaries}
\label{sec:background}

\textbf{Notations.}
For a positive integer $n$, $[n]=\{0,\ldots,n-1\}$. We use lower-case letters with a "hat" symbol for polynomials, bold lower-case for vectors, bold upper-case for matrices, and calligraphic for tensors. Let $R_q=\mathbb{Z}_q[X]/(X^N+1)$ be the negacyclic polynomial ring of degree $N$; for $\hat{x}(X)\in R_q$, $\hat{x}[i]$ is its $i$-th coefficient, and $\rho_\eta(\hat{x})=X^\eta \hat{x} \bmod (X^N+1)$ is a coefficient shift by offset $\eta$. A convolution layer has input shape $(H,W,C_{\mathrm{in}})$, kernel shape $(K,K,C_{\mathrm{in}},C_{\mathrm{out}})$, stride $s_{\mathrm{stride}}$, and padding $p_{\mathrm{pad}}$. Values are fixed-point integers with $f$ fractional bits; a linear product has precision $2f$, and $x\gg\tau$ denotes arithmetic right shift by $\tau$ bits. We write $\langle x\rangle=(\langle x\rangle_C,\langle x\rangle_S)$ for additive sharing over $\mathbb{Z}_{2^\ell}$, with $\langle x\rangle_C+\langle x\rangle_S\equiv x\pmod{2^\ell}$.
A complete notation table is provided in Appendix~\ref{app:notations}.

\subsection{Cryptographic Primitives}
\textbf{Hybrid HE/2PC inference.}
Jaguar follows the standard hybrid interface: HE evaluates linear layers, 2PC evaluates nonlinearities~\cite{juvekar2018gazelle, mishra2020delphi,rathee2020cryptflow2,huang2022cheetah}. We instantiate HE with BFV~\cite{fan2012somewhat}. 
At the HE$\rightarrow$2PC boundary, the server masks an encrypted output $\mathsf{Enc}(\hat{y})$ with $\hat{r}$, sends $\mathsf{Enc}(\hat{y}-\hat{r})$ to the client, and keeps $\hat{r}$ as its share; the client decrypts to obtain the complementary share. Each HE block is converted to shares before the next nonlinear layer, so bootstrapping is unnecessary.


\textbf{Encoding and arithmetic regimes.}
\emph{SIMD encoding}~\cite{juvekar2018gazelle,pang2024bolt,roh2024hyena} 
packs values into transform-domain slots defined by the number-theoretic transform (NTT), the finite-field analogue of the FFT~\cite{pollard1971fast}.
This batching amortizes homomorphic operations but requires a
\emph{prime} plaintext modulus, making the share domain prime-field rather than $\mathbb{Z}_{2^\ell}$ and creating friction with ring-based 2PC.
\emph{Coefficient encoding}~\cite{huang2022cheetah,he2024Rhombus} maps values directly to polynomial coefficients; with careful arrangement, polynomial multiplication realizes convolutions, and Cheetah-style systems~\cite{huang2022cheetah} can avoid costly rotations and use a power-of-two plaintext/share domain. 
\textbf{An important subtlety}: even systems using a power-of-two \emph{plaintext} ring still use a prime-modulus \emph{ciphertext} backend with RNS prime limbs, so plaintext--ciphertext products incur NTT-domain arithmetic and prime-modulus reductions. Jaguar keeps coefficient encoding but changes the ciphertext substrate itself.

\textbf{Threat model.}
Standard two-party private inference: the client holds a private input, the server a proprietary model. The client learns only the final prediction; the server learns nothing beyond what the protocol transcript and output policy reveal. As in prior art~\cite{juvekar2018gazelle, huang2022cheetah, pang2024bolt, hao2022iron}, we assume an honest-but-curious adversary: parties follow the protocol but may infer additional information from messages. Network architecture, tensor shapes, and protocol parameters are public. Security follows by composing the security of BFV, additive secret sharing, and the underlying 2PC protocols.

\subsection{Related Work}

Prior work has improved either the mapping of linear layers to secure computation or the nonlinear protocols (Table~\ref{tab:related_works}).
Gazelle~\cite{juvekar2018gazelle} uses SIMD with a prime plaintext modulus; Cheetah~\cite{huang2022cheetah} introduces coefficient-encoded, rotation-free linear protocols with lean OT-based nonlinearities;
CrypTFlow2~\cite{rathee2020cryptflow2} focuses on faithful 2PC primitives;
Rhombus~\cite{he2024Rhombus} reduces output communication for matrix-vector multiplication via input-output packing; 
Falcon~\cite{xu2023falcon} specializes packing for depthwise convolutions.
Orthogonally, PrivCirNet~\cite{xu2024privcirnet} modifies the model to reduce HE linear cost.

\textbf{Why hasn't power-of-two been tried before?}
Two reasons.
First, the standard FHE library backend (e.g., Microsoft SEAL~\cite{sealcrypto}) is built around NTT-friendly RNS primes;
switching to power-of-two requires re-deriving the convolution kernel because generic plaintext--ciphertext polynomial multiplication via NTT becomes unavailable.
Second, the perceived cost of giving up NTT looks prohibitive \emph{if one assumes generic polynomial multiplication is needed}.
Our \textsc{SPA-Conv} shows that for CNN convolution, scalar--polynomial products and coefficient shifts suffice, and even the residual NTT need (decryption) can be handled locally by the client.
Once both observations land, the prime-modulus rationale weakens and the power-of-two regime becomes the natural choice.

\begin{table}[h]
\centering
\scriptsize
\setlength{\tabcolsep}{3.0pt}
\renewcommand{\arraystretch}{1.15}
\caption{Comparison with representative hybrid HE/2PC private CNN inference systems.}
\vspace{1mm}
\label{tab:related_works}
\begin{tabular}{p{0.13\columnwidth} p{0.14\columnwidth} p{0.25\columnwidth} p{0.15\columnwidth} p{0.08\columnwidth} p{0.16\columnwidth}}
\toprule
\multirow{2}{=}[ -0.7ex ]{\centering Method} & \multicolumn{2}{c}{Linear side} & \multicolumn{1}{c}{Truncation} & \multicolumn{2}{c}{Nonlinear side} \\
\cmidrule(r){2-3}\cmidrule(r){4-4}\cmidrule(){5-6}
& \multicolumn{1}{c}{Target} & \multicolumn{1}{c}{Approach} & \multicolumn{1}{c}{Approach} & \multicolumn{1}{c}{Target} & \multicolumn{1}{c}{Approach} \\
\midrule
Gazelle~\cite{juvekar2018gazelle} & Conv/MatMul & Packed SIMD & -- & ReLU/Pool & Exact GC \\
CrypTFlow2~\cite{rathee2020cryptflow2} & Conv/MatMul & Exact OT/HE backend & Faithful 2PC & Cmp./Div. & Exact 2PC \\
Cheetah~\cite{huang2022cheetah} & Conv/MatMul & Coefficient-encoding protocol & Lean OT & ReLU & Lean OT \\
Rhombus~\cite{he2024Rhombus} & PW-Conv/MatMul & Efficient output repacking & -- & -- & -- \\
Falcon~\cite{xu2023falcon} & DW-Conv & Dense packing & -- & -- & -- \\
\textbf{Jaguar (ours)} & \textbf{Conv} & \textbf{Power-of-two arithmetic regime} & \textbf{HE-side, zero-cost truncation} & \textbf{ReLU} & \textbf{Reduced-bitwidth evaluation} \\
\bottomrule
\end{tabular}
\vspace{1mm}
\end{table}

\section{Proposed Techniques}
\label{sec:proposed}

This section develops Jaguar's contributions in the order they motivate each other.
Section~\ref{sec:zero-cost truncation} starts from the precision-flow mismatch and proves an exact ciphertext-side truncation.
Section~\ref{subsec:power-of-two_regime} shows that the divisibility condition for this truncation, together with $\mathbb{Z}_{2^\ell}$ share-domain alignment, makes a power-of-two ciphertext modulus natural. 
Section~\ref{sec:conv in 2^l} presents \textsc{SPA-Conv}, the convolution kernel for this regime.
Section~\ref{sec:decryption} closes the loop by handling decryption.

\subsection{Zero-Cost Precise Truncation}
\label{sec:zero-cost truncation}

\textbf{Motivation.}
A linear layer multiplies fixed-point inputs and weights with $f$ fractional bits each, producing output at $2f$-bit precision.
Conventional hybrid systems hand this to the ReLU protocol and then run a separate 2PC truncation~\cite{rathee2020cryptflow2,huang2022cheetah}.
Both steps cost interaction.
Our goal: truncate \emph{before} ReLU, with no 2PC interaction.

\textbf{Construction.}
Let $(\hat{c}_0,\hat{c}_1) \in R_q^2$ be a BFV ciphertext encrypting a plaintext $\hat{m}_0$, with $[\hat{c}_0 + \hat{c}_1\hat{s}]_q = \Delta \hat{m}_0+\hat{\epsilon}$ for noise $\hat{\epsilon}$ and scale $\Delta=\lfloor q/p\rfloor$.
Algorithm~\ref{alg:trunc} truncates each ciphertext component by an arithmetic right shift,
implemented as logical shifts on unsigned residues  (Lemma~\ref{lem:unsigned_shift}).
Algorithm~\ref{alg:dec} decrypts the truncated ciphertext using a \emph{two-step} divide-and-round.

\begin{algorithm}[H]
\caption{Proposed truncation}
\label{alg:trunc}
\DontPrintSemicolon
\SetArgSty{textnormal}
\KwIn{
Ciphertext $\llbracket \hat{m}_0\rrbracket=(\hat{c}_0, \hat{c}_1) \in R_q^2$, truncation parameter $f$
}
\KwOut{Ciphertext $\llbracket \hat{m}_1\rrbracket=(\hat{c}'_0, \hat{c}'_1)$ where $\hat{m}_1[i] = (\hat{m}_0[i] \gg f)$}

    \For {$i$ \text{  $\gets 0$ to $N-1$}} {

        $\hat{c}'_0[i] = (\hat{c}_0[i] \gg f)$\;
        $\hat{c}'_1[i] = (\hat{c}_1[i] \gg f)$\;
    }
\end{algorithm}

\begin{algorithm}[H]
\caption{Decryption algorithm for $f$-bit truncated ciphertext}
\label{alg:dec}
\DontPrintSemicolon
\SetArgSty{textnormal}
\KwIn{
Secret key $\mathsf{sk}=\hat{s} \in R_{q}$, ciphertext $\llbracket \hat{m}_1\rrbracket=(\hat{c}'_0, \hat{c}'_1) \in R_{q\gg f}^2$ obtained from Algorithm~\ref{alg:trunc}
}
\KwOut{Plaintext $\hat{m}_1 \in R_{p\gg f}$}

    $\hat{m}_1 \gets \lfloor \frac{1}{2^f} \lfloor \frac{2^f}{\Delta}[\hat{c}'_0 + \hat{c}'_1\hat{s}]_{q\gg f} \rceil \rfloor \;(\textrm{mod} \; {p\gg f})$
\end{algorithm}

\begin{restatable}{theorem}{precisetruncthm}
\label{thm:precise_trunc}
    Algorithm~\ref{alg:trunc} implements exact $f$-bit truncation, decryptable by Algorithm~\ref{alg:dec}, provided that the ciphertext noise satisfies $||\hat{\epsilon}||_\infty < \Delta/2 - (N+1)2^f$. (Proof in Appendix~\ref{app:proof}.)
\end{restatable}
    \textbf{Why two-step divide-and-round.}
    Conventional BFV decryption~\citep{fan2012somewhat} does $\lfloor\frac{1}{\Delta}[\cdot]_{} \rceil$ in one step. 
    Algorithm~\ref{alg:dec} first divides by $\Delta/2^f$, rounds, then truncates by $f$ bits.
    This separation is what makes the result \emph{exact}: 
    because $\Delta$ is divisible by $2^f$ when both $q$ and $p$ are powers of two, the inner divide-by-$\Delta/2^f$ then the rounding cleanly removes the noise polynomial as shown in the proof.
    A single-step decryption applied to a shifted ciphertext would conflate noise with truncation rounding.

    \textbf{Noise budget.}
    Under our parameters $(N, Q, P) = (2048, 54, 30)$ and $f=9$, Theorem~\ref{thm:precise_trunc} requires accumulated noise below $2^{23}-2049\cdot 2^9 \approx7.34\times 10^6$.
    Each Jaguar linear block has multiplicative depth one before share conversion: noise in a fresh ciphertext ($\sigma_\mathrm{in}=3.2$) is scaled by the plaintext filter coefficients and accumulated across \textsc{SPA-Conv} terms.
    Across all three evaluated networks, the worst maximum-accumulation layer is a 960-term pointwise layer in MobileNetV2, which retains a $25.9\sigma$ margin to Theorem~\ref{thm:precise_trunc}'s bound; the corresponding ResNet-18 and ResNet-50 margins are much larger (details in Appendix~\ref{app:noise_budget}).

\subsection{Power-of-Two Arithmetic Regime}
\label{subsec:power-of-two_regime}

\textbf{Modulus choice.}
Security in BFV rests on the concrete hardness of the underlying RLWE instance---ring dimension, modulus, secret distribution, error distribution---and \emph{does not require $q$ to be prime}~\cite{fan2012somewhat,homencstandard,lyubashevsky2013toolkit}.
The prime/RNS convention in HE libraries is implementation-driven, not BFV-driven~\cite{sealmanual}.
Jaguar adopts $q=2^Q$, $p=2^P$, so the BFV scale $\Delta=q/p=2^{Q-P}$ is also a power of two. 
This single choice does three things: 
\textbf{1) Satisfies Theorem~\ref{thm:precise_trunc} by construction;}
    A $f$-bit ciphertext truncation is exact iff $2^f | \Delta$. With $\Delta=2^{Q-P}$, every $f \leq Q-P$ qualifies.
\textbf{2) Aligns HE outputs with the $\mathbb{Z}_{2^\ell}$ share domain;}
    Subsequent 2PC operates over $\mathbb{Z}_{2^\ell}$. Using $p=2^P$ avoids the share-domain mismatch SIMD-encoded systems must bridge.
\textbf{3) Replaces modular reduction with bit masking;}
    Reduction modulo $2^Q$ is a single bitwise AND with $2^Q-1$.

\textbf{Correctness and security.}
The arithmetic backend changes; the privacy goal does not. Correctness follows the standard BFV requirement that accumulated noise stays below $\Delta/2$ minus the truncation slack of Theorem~\ref{thm:precise_trunc};
depth-one linear blocks make this modest. 
Security depends on the concrete RLWE hardness of $(N, Q, \chi_e)$, not on whether $Q$ is the bit-length of a prime.
We choose $(N, Q, P) = (2048, 54, 30)$ to satisfy both correctness and 128-bit classical security (Appendix~\ref{app:security_guarantee}).

\textbf{Algorithmic consequence.}
Jaguar cannot use the standard NTT-friendly prime backend for convolution. The convolution kernel must be expressed in terms the power-of-two ring supports natively---scalar--polynomial products, coefficient shifts, additions, and bit operations---which Sec.~\ref{sec:conv in 2^l} (\textsc{SPA-Conv}) provides.
The single $\hat{c}'_1\cdot\mathsf{sk}$ multiplication needed for decryption is handled by Sec.~\ref{sec:decryption}.

\subsection{\textsc{SPA-Conv}: Scalar--Polynomial Accumulation Convolution}
\label{sec:conv in 2^l}


\begin{figure}[t]
    \centering
    \includegraphics[width=\columnwidth]{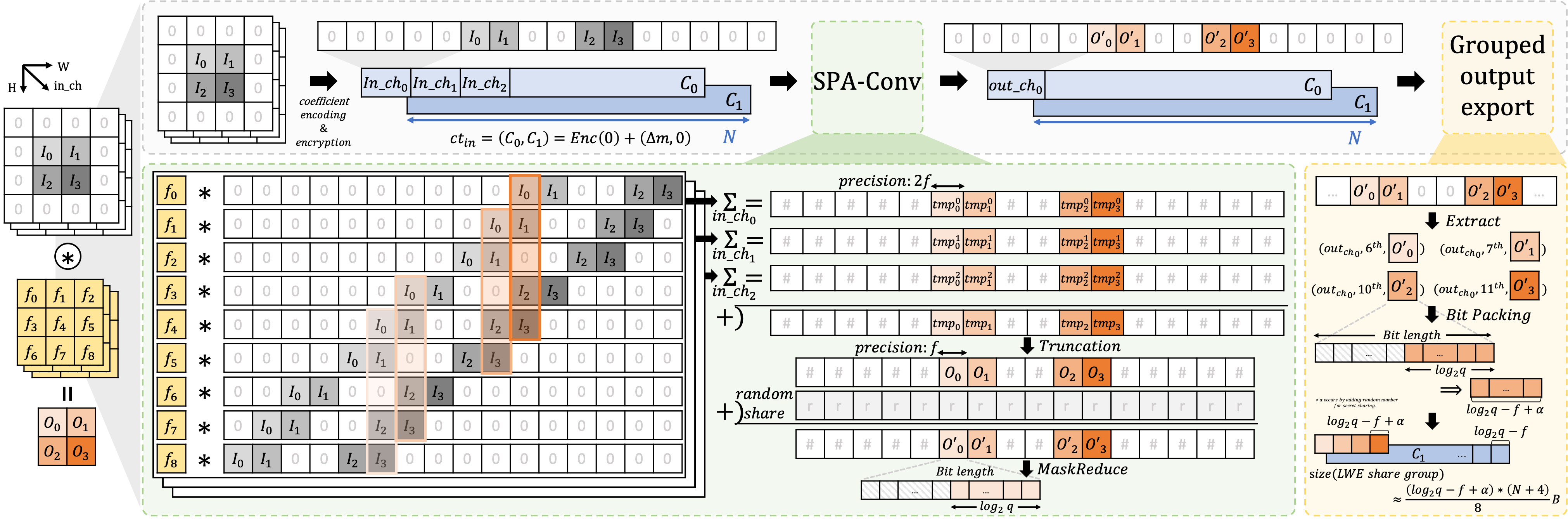}
    \caption{
    Computation flow of \textsc{SPA-Conv}.
    }
    \label{fig:conv}
    \vspace{-1em}
\end{figure}

\textbf{Intuition.}
A convolution is a sum of products of \emph{scalar} filter weights and \emph{spatial} input neighborhoods.
If input channels are coefficient-encoded into RLWE ciphertexts, then for each filter coefficient $w_{o,c,u,v}$, the contribution to output channel $o$ at spatial offset $(u, v)$ is \emph{one scalar} times \emph{one shifted ciphertext}.
Aggregating across input channels and kernel offsets yields the full output.
{No polynomial-by-polynomial multiplication is needed in the linear path; therefore \textbf{no server-side NTT is needed}.}

\textbf{Setup.}
Let $X=\{\mathrm{CT}_t\}_{t=0}^{M-1}$ be encrypted input blocks. A public packing layout $\Lambda(c,j)=(\iota,\sigma)$ returns the source ciphertext index $\iota$ and in-ciphertext coefficient offset $\sigma$ for input channel $c$ and output block $j$, 
covering both \emph{packed} (multiple channels per ciphertext) and \emph{split} (one channel across multiple ciphertexts) layouts. Let $B$ be the number of output blocks per output channel.
For an $K\times K$ kernel with $K=2r+1$, define the spatial-offset set $\Omega_K=\{-r,\ldots,r\}^2$; with padded row width $W_{\mathrm{pp}}$, the coefficient offset for spatial displacement $(u,v)$ is
    $\delta(u,v)=uW_{\mathrm{pp}}+v$.
For output channel $o$ and block $j$, let $\mathcal{C}(o)$ denote the contributing input channels: $\mathcal{C}(o)=[C_{\mathrm{in}}]$ for dense and pointwise convolutions, and $\mathcal{C}(o)=\{o\}$ for depthwise convolution. 

\textbf{Kernel.}
\textsc{SPA-Conv} forms the raw encrypted output block
\[
    A_{o,j}
    =
    \sum_{c\in \mathcal{C}(o)}
    \sum_{(u,v)\in \Omega_K}
    w_{o,c,u,v}\cdot
    \rho_{\sigma(c,j)+\delta(u,v)}
    \!\left(\mathrm{CT}_{\iota(c,j)}\right),
    \qquad
    \Lambda(c,j)=(\iota(c,j),\sigma(c,j)),
\]
where 
$\rho_{\alpha}(\cdot)$ is the negacyclic coefficient shift by offset $\alpha$. 
Each term is a {scalar--polynomial product} (one plaintext scalar applied to both RLWE components), \emph{not} a plaintext--ciphertext polynomial multiplication.
Jaguar then applies ciphertext-side truncation and power-of-two reduction:
\[
    Y_{o,j}=\mathrm{MaskReduce}\!\left(\mathrm{Trunc}_{f}(A_{o,j})\right).
\]
The output crosses the HE$\rightarrow$2PC boundary already at target precision (Section~\ref{sec:jaguar_protocol}). 
The kernel covers all CNN convolution types: pointwise sets $\Omega_K=\{(0,0)\}$ (channel-wise scalar accumulation with no spatial shifts); dense $3\times3$ accumulates nine shifted neighborhoods per input channel; depthwise uses the same shifted accumulation pattern but with $\mathcal{C}(o)=\{o\}$. Detailed algorithms are in Appendix~\ref{app:SPA-Conv_algo}.

\begin{table*}[t]
\caption{
Kernel-level complexity of secure convolution in \emph{estimated 64-bit multiplications}.
Baselines are optimistic lower-envelope estimates under native packing assumptions, using
$\Gamma_{\mathrm{mul}}=6L(\log_2 N+2)$ and
$\Gamma_{\mathrm{rot}}=3\big((\frac{\log_2 N}{2}+2)L^2+(\log_2 N+2)L\big)$; details in Appendix~\ref{app:conv_complexity_derivation}.
$L$ is the number of RNS limbs; $S_{\mathrm{out}}=H'W'$ is the output size; 
$C_w$ is Falcon's per-polynomial channel packing.}
\centering
\footnotesize
\setlength{\tabcolsep}{2.5pt}
\renewcommand{\arraystretch}{0.92}
\setlength{\aboverulesep}{0.35ex}
\setlength{\belowrulesep}{0.35ex}
\begin{tabular*}{\linewidth}{@{\extracolsep{\fill}}ccccc@{}}
\toprule
Category & Method & HE-PMult & HE-Rot & Scalar-Poly \\
\midrule

\multirow{3}{*}{Dense}
& Gazelle
& $HW C_{\mathrm{in}}C_{\mathrm{out}}\Gamma_{\mathrm{mul}}$
& $(HW(C_{\mathrm{in}}+C_{\mathrm{out}})+C_{\mathrm{out}}N)\Gamma_{\mathrm{rot}}$
& $-$ \\

& Cheetah
& $HW C_{\mathrm{in}}C_{\mathrm{out}}\Gamma_{\mathrm{mul}}$
& $0$
& $-$ \\

& \textbf{Jaguar}
& $\mathbf{0}$
& $\mathbf{0}$
& $2NB C_{\mathrm{in}}C_{\mathrm{out}}K^2$ \\

\midrule

\multirow{2}{*}{PW}
& Rhombus-V2
& $S_{\mathrm{out}} C_{\mathrm{in}}C_{\mathrm{out}}\Gamma_{\mathrm{mul}}$
& $\sqrt{S_{\mathrm{out}} C_{\mathrm{in}}C_{\mathrm{out}}N}\Gamma_{\mathrm{rot}}$
& $-$ \\

& \textbf{Jaguar}
& $\mathbf{0}$
& $\mathbf{0}$
& $2NB C_{\mathrm{in}}C_{\mathrm{out}}$ \\

\midrule

\multirow{2}{*}{DW}
& Falcon$^\dagger$
& $\frac{C_{\mathrm{out}}}{C_w}N\Gamma_{\mathrm{mul}}$
& $0$
& $-$ \\

& \textbf{Jaguar}
& $\mathbf{0}$
& $\mathbf{0}$
& $2NB C_{\mathrm{out}}K^2$ \\

\bottomrule
\end{tabular*}
\vspace{0mm}
\label{tab:conv_complexity}
\vspace{-2mm}
\end{table*}

\textbf{Cost.}
Let $T_o=|\mathcal{C}(o)|K^2$. 
Each output block costs $2NT_o$ scalar multiplications (two RLWE components$\times$$N$ coefficients$\times$ $T_o$ shifted scalar terms), plus linear-time shifts, additions, and masks. Summed over $B$ output blocks per channel, \textsc{SPA-Conv} performs $2NBK^2\sum_{o=0}^{C_{\mathrm{out}}-1}|\mathcal{C}(o)|$ coefficient-wise scalar multiplications, scaling as
$\Theta(NBC_{\mathrm{out}}C_{\mathrm{in}}K^2)$ for dense,
$\Theta(NBC_{\mathrm{out}}C_{\mathrm{in}})$ for pointwise, and
$\Theta(NBC_{\mathrm{out}}K^2)$ for depthwise
(Table~\ref{tab:conv_complexity}).
The HE-PMult and HE-Rot columns are zero for Jaguar across all convolution types.
Cost shifts entirely to scalar--polynomial products in a power-of-two ring, where per-operation cost is dominated by integer multiply-and-mask rather than by NTT and modular prime reduction.

\subsection{NTT-Assisted Decryption}
\label{sec:decryption}

A power-of-two $q$ removes NTT from the linear path,
but the client must still compute one polynomial product $\hat{c}'_1\cdot\mathsf{sk}$ during decryption (Algorithm~\ref{alg:dec})---part of BFV decryption itself.
Naively doing this in the power-of-two ring without NTT costs $O(N^2)$, which would dominate at $N=2048$.

\textbf{Construction.}
Jaguar performs $\hat{c}'_1 \cdot \mathsf{sk}$ via NTT \emph{over an auxiliary NTT-friendly} prime $q_{\mathrm{ntt}}$ chosen large enough that the integer-valued coefficients of the product do not wrap modulo $q_{\mathrm{ntt}}$.
The relevant coefficients are then mapped back to residues modulo $2^Q$.
The auxiliary prime never appears in any ciphertext, plaintext, share, or protocol message---it is purely a client-side computational shortcut.

\textbf{Choosing $q_\mathrm{ntt}$.}
With ternary secret keys of Hamming weight $h = \|\mathsf{sk}\|_0 \leq N$, after $f$-bit ciphertext truncation the coefficients of $\hat{c}'_1$ have magnitude below $2^{Q-f}$, so each coefficient of $\hat{c}'_1 \cdot \mathsf{sk}$ is bounded by $h\cdot 2^{Q-f}$.
The sufficient condition is $q_{\mathrm{ntt}}/2 > h\cdot 2^{Q-f}$.
For $N=2048, Q=54, f=9$, we have $h \leq 2^{11}$, giving a bound of $2^{56}$; 
a 60-bit NTT-friendly prime suffices, with margin.

This preserves Jaguar's design:
the ciphertext modulus, share domain, server-side arithmetic, and all messages on the wire remain power-of-two, so Theorem~\ref{thm:precise_trunc} still applies, $\mathbb{Z}_{2^\ell}$ alignment holds, and the server still pays no NTT cost.
We simply use the right tool for the right job---power-of-two arithmetic where it matters for the protocol;
NTT where it matters for asymptotic efficiency of a local product.
Decryption stays $O(N\log N)$.
\section{Jaguar Protocol}
\label{sec:jaguar_protocol}

\vspace{-2mm}
\noindent
\begin{minipage}[t]{0.60\linewidth}
\vspace{0pt}
\subsection{Protocol Overview}
Across layers, activations are maintained as additive shares over $\mathbb{Z}_{2^\ell}$. To enter a linear layer, the client encrypts its share; 
the server runs SPA-CONV on the encrypted share, applies ciphertext-side
truncation, masks the target-precision ciphertexts with fresh random output
shares, and exports additive shares to the 2PC nonlinear protocol.
After the nonlinearity, the client re-encrypts for the next linear layer (Figure~\ref{fig:jaguar_overview}).
\end{minipage}
\hfill
\begin{minipage}[t]{0.40\linewidth}
\vspace{0pt}
\centering
\vspace{-1mm}
\includegraphics[width=0.80\linewidth]{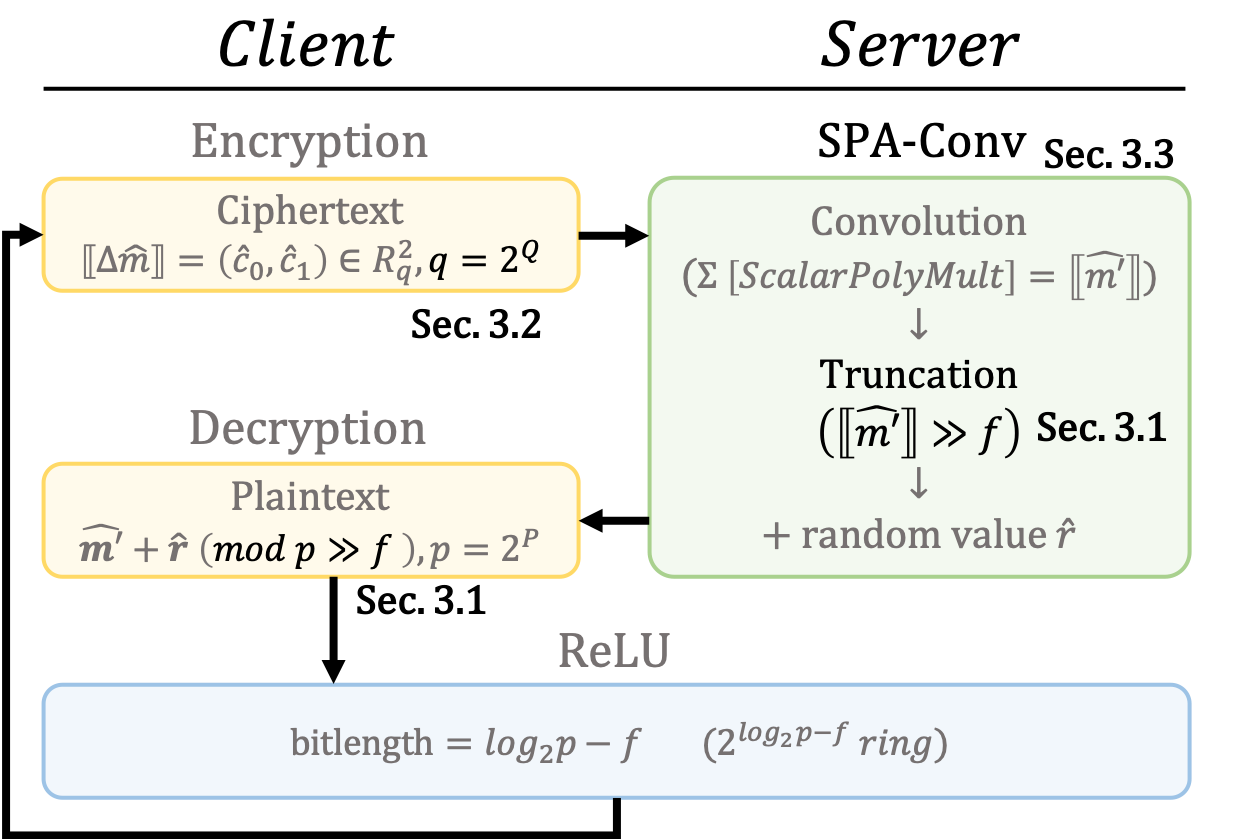}
\vspace{0mm}
\captionof{figure}{Overview of Jaguar.}
\label{fig:jaguar_overview}
\end{minipage}
\vspace{-2mm}

\subsection{Convolution Protocol}
\label{sec:conv_protocol}

Figure~\ref{fig:jaguar_conv_protocol} gives the secure convolution protocol:
the local \textsc{SPA-Conv} kernel wrapped by client encryption before the linear layer and grouped HE-to-share export afterward.

\begin{figure}[t]
\centering
\footnotesize
\setlength{\fboxsep}{3.5pt}
\fbox{
\begin{minipage}{0.97\linewidth}

\textbf{Inputs/outputs.}
($\mathcal{T}$: input tensor, $\mathcal{W}$: weight tensor.)
Client holds $(\langle \mathcal{T}\rangle_C,\mathsf{sk})$; server holds
$(\langle \mathcal{T}\rangle_S,\mathcal{W})$, with
$\langle \mathcal{T}\rangle_C+\langle \mathcal{T}\rangle_S=\mathcal{T}$ over
$\mathbb{Z}_{2^\ell}$.
Output:
$(\langle \mathcal{T}'\rangle_C,\langle \mathcal{T}'\rangle_S)$ for
$\mathcal{T}'=\mathrm{Conv2D}(\mathcal{T},\mathcal{W};s_{\mathrm{stride}})\gg f$.

\textbf{Public.}
$pp=(N,Q,P,K,s_{\mathrm{stride}},p_{\mathrm{pad}})$ and maps
$\Lambda$ and $\Phi$, where $\Lambda$ defines input packing and
$\Phi(o,\xi)=(\gamma,\kappa_{o,\xi})$ identifies the source ciphertext and
coefficient.

\vspace{-0.3em}
\begin{enumerate}
\setlength{\itemsep}{0.05em}
\setlength{\parskip}{0pt}
\setlength{\parsep}{0pt}
\setlength{\topsep}{0.1em}
\setlength{\partopsep}{0pt}
\renewcommand{\labelenumi}{\arabic{enumi}:}

\item \textbf{Client encryption.}
Encode $\langle\mathcal{T}\rangle_C$ by $\Lambda$, encrypt the input blocks, and
send $\{\mathsf{CT}^C_i\}$ to the server.

\item \textbf{Server-side \textsc{SPA-Conv}.}
Evaluate \textsc{SPA-Conv} on $\{\mathsf{CT}^C_i\}$ with plaintext filters.
Inside \textsc{SPA-Conv}, the server adds the server input-share contribution from
$\langle\mathcal{T}\rangle_S$ and $\mathcal{W}$ at pre-truncation precision,
then applies ciphertext-side truncation and power-of-two reduction.
This produces target-precision ciphertexts
$\{\mathsf{CT}^{\mathrm{tr}}_j\}$ encrypting
$\mathrm{Conv2D}(\mathcal{T},\mathcal{W};s_{\mathrm{stride}})\gg f$.
For the first layer, $\langle\mathcal{T}\rangle_S=0$.

\item \textbf{Random-share masking and grouped export.}
For each logical output $(o,\xi)$, sample a output random mask
$r_{o,\xi}\in\mathbb{Z}_{2^{P-f}}$ and homomorphically mask the corresponding
exported constant term so that the client decrypts
$\mathcal{T}'[o,\xi]-r_{o,\xi}$.
Group all valid output coefficients by source ciphertext. Each group contains
one shared linear component $\hat{a}_\gamma$, masked constants
$\{\widetilde{b}_{o,\xi}\}$, coefficient indices $\{\kappa_{o,\xi}\}$, and output
metadata. Send the groups to the client; the server keeps
$r_{o,\xi}$ as its output share.

\item \textbf{Client reconstruction.}
For each group, compute
$\hat{d}_\gamma(X)=\hat{a}_\gamma(X)\mathsf{sk}(X)$ once. For each stored
$\kappa_{o,\xi}$, combine $\hat{d}_\gamma[\kappa_{o,\xi}]$ with
$\widetilde{b}_{o,\xi}$ and decode by power-of-two rounding/shifting to obtain
$\langle\mathcal{T}'\rangle_C[o,\xi]
=\mathcal{T}'[o,\xi]-r_{o,\xi}$.
The server share is
$\langle\mathcal{T}'\rangle_S[o,\xi]=r_{o,\xi}$.
\end{enumerate}
\end{minipage}
}
\vspace{-0.6em}
\caption{Jaguar convolution protocol.}
\label{fig:jaguar_conv_protocol}
\vspace{-1.3em}
\end{figure}

Jaguar adopts Cheetah's~\cite{huang2022cheetah} LWE export compression: LWE ciphertexts extracted from the same RLWE ciphertext share their component, transmitted only once. 
The crucial difference from Cheetah is that exported values are \emph{already at target precision}, so the subsequent nonlinear protocol does not invoke a separate truncation.

\textbf{Fully-connected and attention compatibility.}
Convolution accounts for >\,98\,\% of linear-layer latency in our evaluated CNNs, so our end-to-end implementation focuses on this path and handles the final classifier through a one-time conversion to Cheetah's FC backend~\cite{opencheetah,huang2022cheetah}.
Beyond the CNN workloads evaluated end to end, Jaguar's static-weight linear-to-share arithmetic also applies to attention-based architectures.
In particular, token-wise projections apply the same learned linear transformation independently to each token and are algebraically equivalent to a $1\times1$ convolution~\cite{vaswani2017attention}.
Thus, QKV and output projections can use SPA-CONV and export additive shares to a compatible attention protocol. In contrast, $QK^{\mathsf T}$ and $AV$ have two input-dependent operands, and Softmax is nonlinear; these operations require separate protocols. Section~\ref{sec:attention_extension} evaluates the static-weight kernels, with mapping and composition details in Appendix~\ref{app:limitation}.

\subsection{Nonlinear Protocols}

Jaguar uses the same ring-based 2PC primitives as prior hybrid systems for nonlinearities~\cite{rathee2020cryptflow2,huang2022cheetah}; the difference is the \emph{representation} of their inputs. Because the HE linear path applies exact truncation, the nonlinear protocol receives target-precision shares.
For ReLU, $\mathrm{ReLU}(x)\gg f=\mathrm{ReLU}(x\gg f)$, 
so Jaguar invokes the existing ReLU primitive directly on shares of $x\gg f$.
Table~\ref{tab:nonlinear_comm} quantifies the saving: (a) the post-ReLU truncation protocol---$13\ell$ to $\lambda(\ell+f+2)+19\ell+14f$ communication bits in prior systems---is \emph{eliminated}; 
(b) ReLU communication itself shrinks because the input bitwidth changes from $\ell$ to
$\ell-f$, reducing both comparison width and share payload. 
For MobileNetV2~\cite{sandler2018mobilenetv2}, ReLU6 uses the same target-precision interface with the clipping threshold at the target scale. 

\begin{table}[t]
\centering
\footnotesize
\setlength{\tabcolsep}{3pt}
\renewcommand{\arraystretch}{0.95}
\caption{
Communication costs of (a) truncation and (b) ReLU.
$\ell$: bitwidth of the pre-truncation shared value. $f$: truncation amount. $b_r$: Millionaire radix block size. $\lambda$: security parameter.
}
\label{tab:nonlinear_comm}
\vspace{1mm}

\begin{minipage}[t]{0.49\linewidth}
\centering
\textbf{(a) Truncation}
\vspace{0.5mm}

\begin{adjustbox}{max width=\linewidth}
\begin{tabular}{lccc}
\toprule
Protocol & Small err. & Big err. & Comm. bits \\
\midrule
CrypTFlow2 & $\times$ & $\times$ & $\lambda(\ell+f+2)+19\ell+14f$ \\
Cheetah-F & $\times$ & $\times$ & $16\ell+11f$ \\
Cheetah & $\circ$ & $\times$ & $13\ell$ \\
\textbf{Jaguar} & $\times$ & $\times$ & $\mathbf{0}$ \\
\bottomrule
\end{tabular}
\end{adjustbox}

\vspace{0.5mm}
{\scriptsize Cheetah-F denotes faithful truncation.}
\end{minipage}
\hfill
\begin{minipage}[t]{0.49\linewidth}
\centering
\textbf{(b) ReLU}
\vspace{0.5mm}

\begin{adjustbox}{max width=\linewidth}
\begin{tabular}{lccc}
\toprule
Scheme & Input & Cmp. width & ReLU comm. bits \\
\midrule
Cheetah
& $\ell$
& $b_r\lfloor(\ell-\frac{3f}{2}-1)/b_r\rfloor$
& \makecell[c]{$<11b_r\lfloor(\ell-\frac{3f}{2}-1)/b_r\rfloor$\\$+2\ell+1$} \\

\textbf{Jaguar}
& $\ell-f$
& $b_r\lfloor(\ell-2f-1)/b_r\rfloor$
& \makecell[c]{$<11b_r\lfloor(\ell-2f-1)/b_r\rfloor$\\$+2(\ell-f)+1$} \\
\bottomrule
\end{tabular}
\end{adjustbox}
\end{minipage}

\vspace{-1em}
\end{table}

\section{Experiments}
\label{sec:experiments}

\subsection{Experiment Setup}

\textbf{HE parameters.}
The CNN experiments use BFV~\cite{fan2012somewhat} with $N=2048$, $q=2^{54}$, $p=2^{30}$, $\lambda=128$, $f=9$. 
Parameter selection for preserving fixed-point accuracy is detailed in Appendix~\ref{app:param_selection}.
The operator benchmarks in Section~\ref{sec:attention_extension} use the configurations specified in Appendix~\ref{app:limitation}.

\textbf{Models and dataset.}
ImageNet~\cite{russakovsky2015imagenet} with pretrained torchvision~\cite{torchvision} ResNet-18~\cite{he2016resnet}, ResNet-50~\cite{he2016resnet}, and MobileNetV2~\cite{sandler2018mobilenetv2}. 

\textbf{Implementation.}
Jaguar is built on SEAL~\cite{sealcrypto}. Experiments run on dual Intel Xeon Gold 6250 CPUs at 3.90\,GHz and 128\,GB RAM. 
Linux \texttt{tc} emulates LAN (384\,MB/s, 0.3\,ms RTT) and WAN (62.5\,MB/s, 20\,ms RTT). We compare against OpenCheetah~\cite{opencheetah} and Rhombus~\cite{rhombus_impl} using their public implementations with default HE parameters and matched fixed-point scale and 2PC share bitwidth.
Additional comparisons with Panther~\cite{feng2025panther} and SecONNds~\cite{balla2025seconnds}, including integration with SecONNds' nonlinear protocols, are provided in Appendix~\ref{app:recent-hybrid-systems}.


\subsection{Microbenchmarks}

\begin{table*}[h]
\centering
\small
\setlength{\tabcolsep}{4.5pt}
\renewcommand{\arraystretch}{1.12}
\caption{
Compute-only convolution microbenchmark with AVX disabled.
Speedup is Cheetah over Jaguar. (GM: geometric mean) 
Best/worst cases are shapes $(W,C_{\mathrm{in}},C_{\mathrm{out}},K)$, where $W$ is the square spatial width and $K$ is the kernel size; for depthwise, $C_{\mathrm{out}}=C_{\mathrm{in}}$.}
Full per-shape results in Appendix~\ref{app:full_conv_microbench}.

\label{tab:conv_microbench_avxoff}
\begin{tabular}{llccccc}
\toprule
Conv. type & Layout & \#Cases & GM speedup & Range & Best case & Worst case \\
\midrule
\multirow{2}{*}{Dense $3{\times}3$}
& Packed & 10 & 2.47$\times$ & 0.29--9.64$\times$ & $(32,128,128,3)$ & $(4,512,512,3)$ \\
& Split  & 4  & 17.81$\times$ & 12.81--37.75$\times$ & $(64,64,64,3)$ & $(224,64,64,3)$ \\
\midrule
\multirow{2}{*}{Pointwise $1{\times}1$}
& Packed & 9 & 11.82$\times$ & 3.04--46.73$\times$ & $(28,128,512,1)$ & $(7,320,1280,1)$ \\
& Split  & 4 & 10.76$\times$ & 9.91--12.99$\times$ & $(56,64,256,1)$ & $(112,32,16,1)$ \\
\midrule
\multirow{2}{*}{Depthwise $3{\times}3$}
& Packed & 6 & 11.65$\times$ & 9.64--13.04$\times$ & $(14,96,96,3)$ & $(14,576,576,3)$ \\
& Split  & 3 & 5.45$\times$ & 4.29--6.65$\times$ & $(56,24,24,3)$ & $(112,32,32,3)$ \\
\bottomrule
\end{tabular}

\end{table*}

\textbf{Convolution kernel.}
We benchmark representative convolution shapes from the evaluated models,
measuring server-side \textsc{SPA-Conv}
computation only (single thread, AVX disabled). 

The benefit depends on convolution type and packing amortization (Table~\ref{tab:conv_microbench_avxoff}). 
\textbf{Split dense} (17.81$\times$ GM) gains the most: the $3\times 3$ channel-mixing repeats across many spatial blocks, so the baseline pays NTT/PMult/reduction overhead each time, while \textsc{SPA-Conv} applies block-local shifted scalar accumulation. 
\textbf{Pointwise} (11.82$\times$/10.76$\times$ GM packed/split) is consistently favorable because $K=1$ removes spatial shifts, reducing \textsc{SPA-Conv} to channel-wise scalar accumulation. 
\textbf{Packed depthwise} (11.65$\times$ GM) benefits because each output channel uses only $K^2=9$ shifted terms instead of $C_{\mathrm{in}}K^2$, while the baseline still invokes a generic PMult-heavy path; 
\textbf{split depthwise} (5.45$\times$ GM) is smaller because more spatial blocks must be processed without cross-channel mixing to amortize the repeated block-local work. 
\textbf{Packed dense} (2.47$\times$ GM) shows the widest range: when the spatial footprint is large, only a few channels fit per ciphertext and Jaguar is much faster;
when the footprint is small, the baseline strongly amortizes polynomial multiplication across many packed channels, sometimes reversing the advantage. 
These below-1$\times$ cases do not dominate full-network convolution time.

\begin{table*}[!htp]
\centering
\begin{minipage}[t]{0.58\textwidth}
\vspace{0pt}
\textbf{Nonlinear stage.}
Table~\ref{tab:nonlinear_breakdown_rn50} measures the ReLU and truncation layers from ResNet-50. ReLU latency drops by 2.24$\times$ under LAN and 1.61$\times$ under WAN, with 1.83$\times$ lower communication, because the ReLU comparison runs at a smaller input bitwidth.
The truncation gain is larger: replacing the interactive 2PC truncation with local HE-side truncation gives 5.58$\times$ lower latency under LAN and 16.22$\times$ under WAN (the WAN gain is dominated by the eliminated round trip), while reducing truncation communication to zero. The combined ReLU+truncation stage improves by 2.32$\times$(LAN) / 1.80$\times$(WAN) in latency and 1.98$\times$ in communication.
\end{minipage}
\hfill
\begin{minipage}[t]{0.40\textwidth}
\vspace{0pt}
\centering
\scriptsize
\setlength{\tabcolsep}{3.5pt}
\renewcommand{\arraystretch}{1.10}
\caption{
Network-level nonlinear-stage breakdown for ResNet-50.
}
\label{tab:nonlinear_breakdown_rn50}
\begin{tabular}{ccccc}
\toprule
\multirow{2}{*}[-0.5ex]{Layer} & \multirow{2}{*}[-0.5ex]{Method}
& \multicolumn{2}{c}{Latency (s)} & \multirow{2}{*}[-0.5ex]{Comm. (MiB)} \\
\cmidrule(l){3-4}
& & LAN & WAN & \\
\midrule
\multirow{3}{*}{ReLU}
& Cheetah & 17.70 & 24.20 & 271.69 \\
& \textbf{Jaguar} & \textbf{7.89} & \textbf{15.02} & \textbf{148.61} \\
\cmidrule(){3-4}\cmidrule(l){5-5}
& & 2.24$\times$ & 1.61$\times$ & 1.83$\times$ \\
\midrule
\multirow{3}{*}{Truncation}
& Cheetah & 1.08 & 3.21 & 22.91 \\
& \textbf{Jaguar} & \textbf{0.19} & \textbf{0.20} & \textbf{0.00} \\
\cmidrule(){3-4}\cmidrule(l){5-5}
& & 5.58$\times$ & 16.22$\times$ & Eliminated \\
\bottomrule
\end{tabular}

\end{minipage}
\end{table*}

\subsection{End-to-End Inference}
\label{subsec:end2end}
\begin{table*}[h]
\centering
\footnotesize
\setlength{\tabcolsep}{7pt}
\renewcommand{\arraystretch}{0.98}
\setlength{\aboverulesep}{0.25ex}
\setlength{\belowrulesep}{0.25ex}
\caption{
End-to-end latency and communication with AVX disabled.
Speedup and communication reduction are computed relative to Cheetah.
(AVX-enabled results are in Appendix~\ref{app:avx}.)
}
\label{tab:e2e_avxoff}
\vspace{0mm}
\begin{tabular}{@{}cccccccc@{}}
\toprule
\multirow{2}{*}[-0.5ex]{Model} &
\multirow{2}{*}[-0.5ex]{Protocol} &
\multicolumn{2}{c}{LAN} &
\multicolumn{2}{c}{WAN} &
\multicolumn{2}{c}{Communication} \\
\cmidrule(lr){3-4}
\cmidrule(lr){5-6}
\cmidrule(l){7-8}
& & Latency (s) & Speedup & Latency (s) & Speedup & Comm. (MiB) & Reduction \\
\midrule

\multirow{3}{*}{ResNet18}
& Cheetah & 30.30 & 1.00$\times$ & 36.94 & 1.00$\times$ & 366.94 & 1.00$\times$ \\
& Rhombus & 29.04 & 1.04$\times$ & 36.24 & 1.02$\times$ & 344.01 & 1.07$\times$ \\
& \textbf{Jaguar} & \textbf{8.64} & \textbf{3.51}$\times$ & \textbf{15.94} & \textbf{2.32}$\times$ & \textbf{315.02} & \textbf{1.16}$\times$ \\
\midrule

\multirow{3}{*}{ResNet50}
& Cheetah & 120.82 & 1.00$\times$ & 142.72 & 1.00$\times$ & 1335.63 & 1.00$\times$ \\
& Rhombus & 96.77 & 1.25$\times$ & 111.44 & 1.28$\times$ & \textbf{750.22} & \textbf{1.78}$\times$ \\
& \textbf{Jaguar} & \textbf{32.44} & \textbf{3.72}$\times$ & \textbf{51.50} & \textbf{2.77}$\times$ & 962.00 & 1.39$\times$ \\
\midrule

\multirow{3}{*}{MobileNetV2}
& Cheetah & 46.99 & 1.00$\times$ & 67.25 & 1.00$\times$ & 1042.62 & 1.00$\times$ \\
& Rhombus & 58.98 & 0.80$\times$ & 77.19 & 0.87$\times$ & 836.81 & 1.25$\times$ \\
& \textbf{Jaguar} & \textbf{18.41} & \textbf{2.55}$\times$ & \textbf{32.42} & \textbf{2.07}$\times$ & \textbf{592.50} & \textbf{1.76}$\times$ \\
\bottomrule
\end{tabular}
\end{table*}

\begin{figure}[h]
    \centering
    \includegraphics[width=0.95\columnwidth]{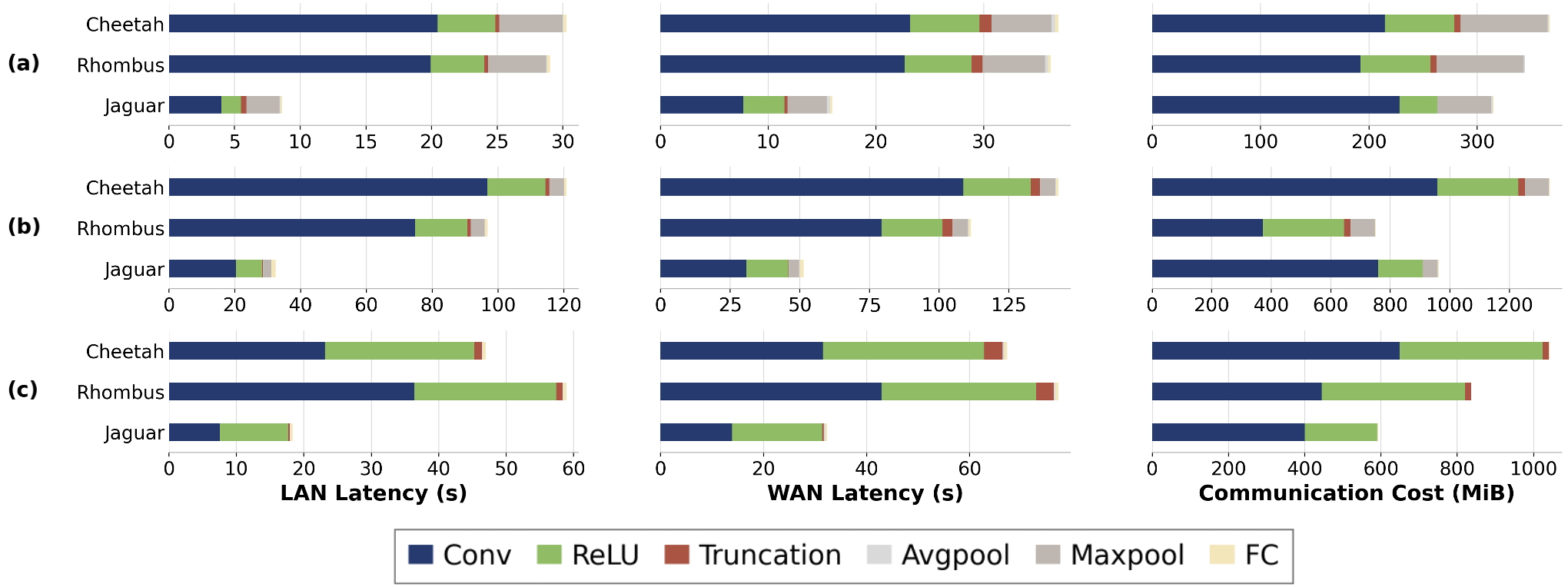}
    \caption{
    End-to-end latency and communication breakdown across ImageNet-scale CNNs:
    (a) ResNet-18, (b) ResNet-50, and (c) MobileNetV2.
    }
    \label{fig:e2e_breakdown_grid}
    \vspace{-1em}
\end{figure}

We evaluate end-to-end secure inference using four threads (Table~\ref{tab:e2e_avxoff}).
Latency includes local protocol execution and network transfer; communication counts all messages exchanged in one inference.

Jaguar achieves the lowest end-to-end latency on every model:
3.51$\times$/3.72$\times$/2.55$\times$ over Cheetah on LAN; 2.32$\times$/2.77$\times$/2.07$\times$ on WAN; 
communication 1.16--1.76$\times$ lower than Cheetah.
The improvement holds across CNNs with different mixtures of dense, pointwise, and depthwise convolutions.

\textbf{ResNet-50 communication vs. Rhombus.}
Rhombus reports 750\,MiB on ResNet-50, lower than Jaguar's 962\,MiB.
Rhombus is \emph{specifically optimized} for the pointwise-convolution and MVM components that dominate ResNet-50's communication;
its integration replaces these in Cheetah but leaves dense 3$\times$3 convolution and the nonlinear stage untouched.
This is why Rhombus's communication advantage on ResNet-50 is largest in absolute terms but \textbf{does not translate proportionally into latency}:
Rhombus is 1.25$\times$ faster than Cheetah on LAN; Jaguar is 3.72$\times$.
Jaguar trades a portion of communication on ResNet-50 for substantially lower compute across the \emph{entire} convolution path and a fundamentally cheaper nonlinear stage.
On more diverse layer mixes (ResNet-18, MobileNetV2), Jaguar achieves both lower latency \emph{and} lower communication than Rhombus.

\textbf{End-to-end breakdown.}
Across all models, the dominant latency reduction comes from the linear convolution path,
with the relative contribution differing by architecture (Figure~\ref{fig:e2e_breakdown_grid}). 
The nonlinear side adds 
reduced ReLU input bitwidth and eliminated post-ReLU truncation,
improving both compute-dominated LAN and communication-sensitive WAN latency.

\textbf{AVX-512.}
Jaguar's power-of-two arithmetic is well-suited to AVX-512~\cite{intel_avx512_overview,intel_intrinsics_guide}, since modular reduction and ciphertext-side truncation map to vectorized bit masking and arithmetic shifts. With AVX-512, Jaguar's end-to-end LAN latencies reach 7.90\,s/25.55\,s/16.83\,s on ResNet-18/ResNet-50/MobileNetV2, respectively (Appendix~\ref{app:avx}).

\subsection{Extensions to Attention Architectures}
\label{sec:attention_extension}

To assess this mapping beyond CNN workloads, we benchmark the static-weight projections surrounding the input-dependent attention core. Fused QKV projections produce queries, keys, and values before attention, while output projections transform the aggregated values afterward; both apply fixed channel-mixing weights independently to each token.

We evaluate eight projection shapes covering 44 fused QKV/output-projection occurrences in MaxViT-T~\cite{tu2022maxvit}.
With matched 37-bit plaintext arithmetic, nine fractional bits, and server-side HE linear timing, Jaguar achieves a $2.24\times$ occurrence-weighted speedup over Cheetah; individual ratios range from $0.84\times$ to $11.03\times$, with six of eight shapes faster.
Sparse BERT-base projections show $2.15$--$2.94\times$ speedup at 1\% weight density. These are kernel benchmarks, not end-to-end attention measurements; $QK^{\mathsf T}$, Softmax, and $AV$ remain outside SPA-CONV. Appendix~\ref{app:limitation} reports the setup, per-shape results, and share-level composition.

\section{Conclusion}
\label{sec:conclusion}

We presented Jaguar, a hybrid HE/2PC private neural network inference system built on a power-of-two ciphertext ring. 
The single design choice $q=2^Q, p=2^P$ unlocks three contributions that prior work could not realize on the conventional NTT-prime backend:
\textsc{SPA-Conv}, which evaluates convolution as scalar--polynomial accumulation without polynomial multiplication, NTT, or modular prime reduction in the linear path;
exact pre-ReLU ciphertext-side truncation, which removes the post-ReLU truncation protocol and reduces ReLU input bitwidth;
and NTT-assisted decryption, which recovers $O(N\log N)$ client reconstruction via an auxiliary NTT prime without disturbing the protocol substrate.
Experiments on ResNet-18, ResNet-50, and MobileNetV2 show 2.07--3.72$\times$ lower end-to-end latency than Cheetah and 2.16--3.36$\times$ lower than Rhombus under both LAN and WAN settings. 
Jaguar shows that \emph{changing the HE arithmetic regime itself}---rather than only optimizing packing or communication on top of the conventional backend---can unlock substantial gains for practical private neural network inference.

\subsubsection*{Acknowledgments}
This work was supported in part by the Institute of Information \& Communications Technology Planning \& Evaluation (IITP) grant funded by the Korean government (MSIT) (No. RS-2025-02217656, Privacy Preserving AI Integrated System for K-Cloud based on Homomorphic Encryption), by K-CHIPS (Korea Collaborative \& High-tech Initiative for Prospective Semiconductor Research) (2410018409, RS-2026-25524122, 26086-45FC) funded by the Ministry of Trade, Industry \& Energy (MOTIE, Korea), and by IITP under the Graduate School of Artificial Intelligence Semiconductor (IITP-2026-RS-2023-00256081) grant funded by the Korea government (MSIT).



\newpage
\bibliographystyle{unsrtnat}
\bibliography{ref}

@article{roh2024flash,
  title={{Flash}: A Hybrid Private Inference Protocol for Deep {CNNs} with High Accuracy and Low Latency on CPU},
  author={Roh, Hyeri and Yeo, Jinsu and Ko, Yeongil and Wei, Gu-Yeon and Brooks, David and Choi, Woo-Seok},
  journal={arXiv preprint arXiv:2401.16732},
  year={2024}
}

@inproceedings{huang2022cheetah,
  title={Cheetah: Lean and fast secure {Two-Party} deep neural network inference},
  author={Huang, Zhicong and Lu, Wen-jie and Hong, Cheng and Ding, Jiansheng},
  booktitle={31st USENIX Security Symposium (USENIX Security 22)},
  pages={809--826},
  year={2022}
}

@inproceedings{rathee2020cryptflow2,
  title={Cryptflow2: Practical 2-party secure inference},
  author={Rathee, Deevashwer and Rathee, Mayank and Kumar, Nishant and Chandran, Nishanth and Gupta, Divya and Rastogi, Aseem and Sharma, Rahul},
  booktitle={Proceedings of the 2020 ACM SIGSAC Conference on Computer and Communications Security},
  pages={325--342},
  year={2020}
}

@article{choi2022impala,
  title={Impala: Low-latency, communication-efficient private deep learning inference},
  author={Choi, Woo-Seok and Reagen, Brandon and Wei, Gu-Yeon and Brooks, David},
  journal={arXiv preprint arXiv:2205.06437},
  year={2022}
}

@inproceedings{roh2024hyena,
  title={Hyena: Optimizing Homomorphically Encrypted Convolution for Private CNN Inference},
  author={Roh, Hyeri and Choi, Woo-Seok},
  booktitle={Proceedings of the 43rd IEEE/ACM International Conference on Computer-Aided Design},
  pages={1--9},
  year={2024}
}

@inproceedings{juvekar2018gazelle,
  title={{GAZELLE}: A low latency framework for secure neural network inference},
  author={Juvekar, Chiraag and Vaikuntanathan, Vinod and Chandrakasan, Anantha},
  booktitle={27th USENIX Security symposium (USENIX Security 18)},
  pages={1651--1669},
  year={2018}
}

@article{fan2012somewhat,
  title={Somewhat practical fully homomorphic encryption},
  author={Fan, Junfeng and Vercauteren, Frederik},
  journal={Cryptology ePrint Archive},
  year={2012},
  url = {https://eprint.iacr.org/2012/144}
}

@inproceedings{he2024Rhombus,
    author = {He, Jiaxing and Yang, Kang and Tang, Guofeng and Huang, Zhangjie and Lin, Li and Wei, Changzheng and Yan, Ying and Wang, Wei},
    title = {Rhombus: Fast Homomorphic Matrix-Vector Multiplication for Secure Two-Party Inference},
    booktitle = {Proceedings of the 2024 on ACM SIGSAC Conference on Computer and Communications Security (CCS)},
    pages = {2490––2504},
    year = {2024}
}

@inproceedings{xu2023falcon,
  author    = {Tianshi Xu and Meng Li and Runsheng Wang and Ru Huang},
  title     = {Falcon: Accelerating Homomorphically Encrypted Convolutions for Efficient Private Mobile Network Inference},
  booktitle = {2023 IEEE/ACM International Conference on Computer-Aided Design ({ICCAD})},
  year      = {2023},
  pages     = {1--9},
  doi       = {10.1109/ICCAD57390.2023.10323672}
}

@inproceedings{hao2022iron,
  author    = {Meng Hao and Hongwei Li and Hanxiao Chen and Pengzhi Xing and Guowen Xu and Tianwei Zhang},
  title     = {Iron: Private Inference on Transformers},
  booktitle = {Advances in Neural Information Processing Systems 35},
  year      = {2022},
  pages     = {15718--15731}
}

@inproceedings{pang2024bolt,
  author    = {Qi Pang and Jinhao Zhu and Helen M{\"o}llering and Wenting Zheng and Thomas Schneider},
  title     = {{BOLT}: Privacy-Preserving, Accurate and Efficient Inference for Transformers},
  booktitle = {2024 IEEE Symposium on Security and Privacy ({SP})},
  year      = {2024},
  pages     = {4753--4771},
  doi       = {10.1109/SP54263.2024.00130}
}

@inproceedings{xu2024privcirnet,
  author    = {Tianshi Xu and Lemeng Wu and Runsheng Wang and Meng Li},
  title     = {PrivCirNet: Efficient Private Inference via Block Circulant Transformation},
  booktitle = {Advances in Neural Information Processing Systems 37},
  year      = {2024}
}

@inproceedings{mohassel2017secureml,
  title={SecureML: A System for Scalable Privacy-Preserving Machine Learning},
  author={Mohassel, Payman and Zhang, Yupeng},
  booktitle={2017 IEEE Symposium on Security and Privacy (SP)},
  pages={19--38},
  year={2017},
  organization={IEEE},
  doi={10.1109/SP.2017.12}
}

@inproceedings{liu2017minionn,
  title={Oblivious Neural Network Predictions via MiniONN Transformations},
  author={Liu, Jian and Juuti, Mika and Lu, Yao and Asokan, N.},
  booktitle={Proceedings of the 2017 ACM SIGSAC Conference on Computer and Communications Security},
  pages={619--631},
  year={2017},
  doi={10.1145/3133956.3134056}
}

@inproceedings{mishra2020delphi,
  title={DELPHI: A Cryptographic Inference Service for Neural Networks},
  author={Mishra, Pratyush and Lehmkuhl, Ryan and Srinivasan, Akshayaram and Zheng, Wenting and Popa, Raluca Ada},
  booktitle={29th USENIX Security Symposium (USENIX Security 20)},
  pages={2505--2522},
  year={2020}
}

@inproceedings{giladbachrach2016cryptonets,
  title={CryptoNets: Applying Neural Networks to Encrypted Data with High Throughput and Accuracy},
  author={Gilad-Bachrach, Ran and Dowlin, Nathan and Laine, Kim and Lauter, Kristin and Naehrig, Michael and Wernsing, John},
  booktitle={Proceedings of the 33rd International Conference on Machine Learning},
  series={Proceedings of Machine Learning Research},
  volume={48},
  pages={201--210},
  year={2016},
  publisher={PMLR}
}

@inproceedings{brutzkus2019lola,
  title={Low Latency Privacy Preserving Inference},
  author={Brutzkus, Alon and Elisha, Oren and Gilad-Bachrach, Ran},
  booktitle={Proceedings of the 36th International Conference on Machine Learning},
  series={Proceedings of Machine Learning Research},
  volume={97},
  pages={812--821},
  year={2019},
  publisher={PMLR}
}

@inproceedings{lee2022lowcomplexity,
  title={Low-Complexity Deep Convolutional Neural Networks on Fully Homomorphic Encryption Using Multiplexed Parallel Convolutions},
  author={Lee, Eunsang and Lee, Joon-Woo and Lee, Junghyun and Kim, Young-Sik and Kim, Yongjune and No, Jong-Seon and Choi, Woosuk},
  booktitle={Proceedings of the 39th International Conference on Machine Learning},
  series={Proceedings of Machine Learning Research},
  volume={162},
  pages={12403--12422},
  year={2022},
  publisher={PMLR}
}

@inproceedings{ju2024neujeans,
  title={NeuJeans: Private Neural Network Inference with Joint Optimization of Convolution and FHE Bootstrapping},
  author={Ju, Jae Hyung and Park, Jaiyoung and Kim, Jongmin and Kang, Minsik and Kim, Donghwan and Cheon, Jung Hee and Ahn, Jung Ho},
  booktitle={Proceedings of the 2024 ACM SIGSAC Conference on Computer and Communications Security},
  pages={4361--4375},
  year={2024},
  doi={10.1145/3658644.3690375}
}

@misc{stoian2023tfhe,
  title={Deep Neural Networks for Encrypted Inference with TFHE},
  author={Stoian, Alexandru and Fr{\'e}ry, Jordan and Bredehoft, Roman and Montero, Luis and Kherfallah, Celia and Chevallier-Mames, Beno{\^i}t},
  howpublished={Cryptology ePrint Archive, Paper 2023/257},
  year={2023},
  url={https://eprint.iacr.org/2023/257}
}

@article{gupta2022llama,
  title={{LLAMA}: A Low Latency Math Library for Secure Inference},
  author={Gupta, Kanav and Kumaraswamy, Deepak and Chandran, Nishanth and Gupta, Divya},
  journal={Proceedings on Privacy Enhancing Technologies},
  volume={2022},
  number={4},
  pages={274--294},
  year={2022}
}

@inproceedings{he2016resnet,
  title={Deep Residual Learning for Image Recognition},
  author={He, Kaiming and Zhang, Xiangyu and Ren, Shaoqing and Sun, Jian},
  booktitle={Proceedings of the IEEE Conference on Computer Vision and Pattern Recognition},
  pages={770--778},
  year={2016},
  doi={10.1109/CVPR.2016.90}
}

@inproceedings{sandler2018mobilenetv2,
  title={{MobileNetV2}: Inverted Residuals and Linear Bottlenecks},
  author={Sandler, Mark and Howard, Andrew and Zhu, Menglong and Zhmoginov, Andrey and Chen, Liang-Chieh},
  booktitle={Proceedings of the IEEE Conference on Computer Vision and Pattern Recognition},
  pages={4510--4520},
  year={2018},
  doi={10.1109/CVPR.2018.00474}
}

@article{russakovsky2015imagenet,
  title={ImageNet Large Scale Visual Recognition Challenge},
  author={Russakovsky, Olga and Deng, Jia and Su, Hao and Krause, Jonathan and Satheesh, Sanjeev and Ma, Sean and Huang, Zhiheng and Karpathy, Andrej and Khosla, Aditya and Bernstein, Michael and Berg, Alexander C. and Li, Fei-Fei},
  journal={International Journal of Computer Vision},
  volume={115},
  number={3},
  pages={211--252},
  year={2015},
  doi={10.1007/s11263-015-0816-y}
}

@article{esteva2017dermatologist,
  title={Dermatologist-Level Classification of Skin Cancer with Deep Neural Networks},
  author={Esteva, Andre and Kuprel, Brett and Novoa, Roberto A. and Ko, Justin and Swetter, Susan M. and Blau, Helen M. and Thrun, Sebastian},
  journal={Nature},
  volume={542},
  number={7639},
  pages={115--118},
  year={2017},
  doi={10.1038/nature21056}
}

@article{kaissis2021end,
  title={End-to-End Privacy Preserving Deep Learning on Multi-Institutional Medical Imaging},
  author={Kaissis, Georgios A. and Ziller, Alexander and Passerat-Palmbach, Jonathan and Ryffel, Th{\'e}o and Usynin, Dmitrii and Trask, Andrew and Lima, Ion{\'e}sio and Mancuso, Jason and Jungmann, Friederike and Steinborn, Marc-Matthias and Braren, Rickmer and Makowski, Marcus and Rueckert, Daniel and others},
  journal={Nature Machine Intelligence},
  volume={3},
  number={6},
  pages={473--484},
  year={2021},
  doi={10.1038/s42256-021-00337-8}
}

@inproceedings{lin2020mcunet,
  title={{MCUNet}: Tiny Deep Learning on {IoT} Devices},
  author={Lin, Ji and Chen, Wei-Ming and Lin, Yujun and Gan, Chuang and Han, Song},
  booktitle={Advances in Neural Information Processing Systems},
  volume={33},
  pages={11711--11722},
  year={2020}
}

@book{warden2019tinyml,
  title={{TinyML}: Machine Learning with TensorFlow Lite on Arduino and Ultra-Low-Power Microcontrollers},
  author={Warden, Pete and Situnayake, Daniel},
  publisher={O'Reilly Media},
  year={2019},
  isbn={9781492052043}
}

@misc{opencheetah,
  title        = {{OpenCheetah}: Proof-of-Concept Implementation for {Cheetah}},
  author       = {{Alibaba Gemini Lab}},
  howpublished = {\url{https://github.com/Alibaba-Gemini-Lab/OpenCheetah}},
  year         = {2022},
  note         = {Accessed: 2026-04-26}
}

@misc{homencstandard,
  title        = {Homomorphic Encryption Standard},
  author       = {{HomomorphicEncryption.org Standardization Consortium}},
  howpublished = {Version 1.1},
  year         = {2024},
  url          = {https://homomorphicencryption.org/wp-content/uploads/2024/08/Homomorphic-Encryption-Standard-v1.1.pdf}
}

@manual{sealmanual,
  title        = {Simple Encrypted Arithmetic Library 2.3.1},
  author       = {Laine, Kim},
  organization = {Microsoft Research},
  year         = {2017},
  url          = {https://www.microsoft.com/en-us/research/wp-content/uploads/2017/11/sealmanual-2-3-1.pdf}
}

@inproceedings{lyubashevsky2013toolkit,
  title     = {A Toolkit for Ring-{LWE} Cryptography},
  author    = {Lyubashevsky, Vadim and Peikert, Chris and Regev, Oded},
  booktitle = {Advances in Cryptology -- EUROCRYPT 2013},
  series    = {Lecture Notes in Computer Science},
  volume    = {7881},
  pages     = {35--54},
  year      = {2013},
  publisher = {Springer},
  doi       = {10.1007/978-3-642-38348-9_3}
}

@inproceedings{vaswani2017attention,
  title     = {Attention Is All You Need},
  author    = {Vaswani, Ashish and Shazeer, Noam and Parmar, Niki and Uszkoreit, Jakob and Jones, Llion and Gomez, Aidan N. and Kaiser, Lukasz and Polosukhin, Illia},
  booktitle = {Advances in Neural Information Processing Systems},
  volume    = {30},
  year      = {2017}
}

@misc{intel_avx512_overview,
  title        = {Intel Advanced Vector Extensions 512 ({Intel AVX-512}) Overview},
  author       = {{Intel Corporation}},
  howpublished = {\url{https://www.intel.com/content/www/us/en/architecture-and-technology/avx-512-overview.html}},
  year         = {2024},
  note         = {Accessed: 2026-04-27}
}

@misc{intel_intrinsics_guide,
  title        = {Intel Intrinsics Guide},
  author       = {{Intel Corporation}},
  howpublished = {\url{https://www.intel.com/content/www/us/en/docs/intrinsics-guide/index.html}},
  year         = {2024},
  note         = {Version 3.6.9, accessed: 2026-04-27}
}

@misc{sealcrypto,
  title        = {{M}icrosoft {SEAL} (release 4.1)},
  howpublished = {\url{https://github.com/Microsoft/SEAL}},
  month        = jan,
  year         = {2023},
  note         = {Microsoft Research, Redmond, WA.},
  key          = {SEAL}
}

@misc{rhombus_impl,
  title        = {{RhombusEnd2End}: Public Implementation of Rhombus},
  author       = {He, Jiaxing},
  year         = {2025},
  howpublished ={\url{https://github.com/2646jx/RhombusEnd2End}},
  note         = {MIT License. Accessed: 2026-04-30}
}

@article{pollard1971fast,
  title={The Fast Fourier Transform in a Finite Field},
  author={Pollard, J. M.},
  journal={Mathematics of Computation},
  volume={25},
  number={114},
  pages={365--374},
  year={1971}
}

@misc{torchvision,
  title        = {Torchvision: PyTorch's Computer Vision Library},
  author       = {{PyTorch Contributors}},
  howpublished = {\url{https://github.com/pytorch/vision}},
  year         = {2024},
  note         = {Accessed: 2026-04-30}
}

@article{feng2025panther,
  author  = {Jun Feng and Yefan Wu and Hong Sun and Shunli Zhang and Debin Liu},
  title   = {{Panther}: Practical Secure Two-Party Neural Network Inference},
  journal = {IEEE Transactions on Information Forensics and Security},
  volume  = {20},
  pages   = {1149--1162},
  year    = {2025},
  doi     = {10.1109/TIFS.2025.3526063}
}

@article{balla2025seconnds,
  author  = {Shashank Balla},
  title   = {{SecONNds}: Secure Outsourced Neural Network Inference on {ImageNet}},
  journal = {arXiv preprint arXiv:2506.11586},
  year    = {2025},
  doi     = {10.48550/arXiv.2506.11586}
}

@misc{seconnds_code,
  author       = {Shashank Balla},
  title        = {{SecONNds}: Open-Source Implementation},
  howpublished = {\url{https://github.com/shashankballa/SecONNds}}
}

@inproceedings{tu2022maxvit,
  title     = {{MaxViT}: Multi-Axis Vision Transformer},
  author    = {Tu, Zhengzhong and Talebi, Hossein and Zhang, Han
               and Yang, Feng and Milanfar, Peyman and Bovik, Alan
               and Li, Yinxiao},
  booktitle = {Computer Vision -- ECCV 2022},
  series    = {Lecture Notes in Computer Science},
  volume    = {13684},
  pages     = {459--479},
  year      = {2022},
  publisher = {Springer},
  doi       = {10.1007/978-3-031-20053-3_27}
}

\newpage
\appendix
\section{Notations}
\label{app:notations}

\begingroup
\small
\setlength{\LTpre}{2pt}
\setlength{\LTpost}{2pt}
\setlength{\tabcolsep}{2.5pt}
\renewcommand{\arraystretch}{1.06}

\noindent\textbf{Conventions.}
Ring elements are italic polynomials, e.g., $\hat{x}(X)\in R_q$ or short for $\hat{x}\in R_q$, and $\hat{x}[i]$ denotes the $i$-th coefficient. Lower-case letters with "hat" symbols denote polynomials, bold lower-case symbols denote vectors/lists, bold upper-case symbols denote matrices, and calligraphic symbols denote tensors.

\vspace{1mm}
\begin{longtable}{@{}L{0.14\textwidth}L{0.31\textwidth}@{\hspace{0.7em}\vrule width 0.35pt\hspace{0.7em}}L{0.14\textwidth}L{0.31\textwidth}@{}}
\toprule
Symbol & Meaning & Symbol & Meaning \\
\midrule
\multicolumn{4}{@{}l}{\textbf{General and cryptographic parameters}} \\
\midrule
$[n]$ & $\{0,\ldots,n-1\}$ for an integer $n$.
&
$N$ & Ring degree / number of polynomial coefficients. \\

$R_q$ & $\mathbb{Z}_q[X]/(X^N+1)$.
&
$q=2^Q$ & Ciphertext modulus and bitwidth. \\

$p=2^P$ & Plaintext modulus and bitwidth.
&
$\Delta$ & BFV scale, $\Delta=q/p=2^{Q-P}$. \\


$q_{\mathrm{ntt}}$ & Auxiliary NTT-friendly prime for client-side reconstruction.
&
$\lambda$ & Security parameter. \\

$\ell$ & Bitwidth of the additive-share ring for the pre-truncation value; i.e. shares are over $\mathbb{Z}_{2^\ell}$.
&
$f$ & Fixed-point fractional bits; main truncation amount. \\

$b_r$ & Millionaire comparison radix block size. 
&
$L$ & Number of RNS limbs in prime-modulus baselines. \\ 

$\Gamma_{\mathrm{mul}}$ & Estimated 64-bit multiplication cost for HE-PMult. 
&
$\Gamma_{\mathrm{rot}}$ & Estimated 64-bit multiplication cost for HE-Rot. \\

$h$ & Hamming weight of the ternary secret key. \\

\midrule
\multicolumn{4}{@{}l}{\textbf{Ring elements, ciphertexts, and shares}} \\
\midrule
$\hat{x}(X)\in R_q$ & Ring element / polynomial.
&
$\hat{x}[i]$ & $i$-th coefficient of $x(X)$. \\

$\rho_\eta(\hat{x})$ & Negacyclic shift $X^\eta \hat{x} \bmod (X^N+1)$.
&
$\mathsf{sk}=\hat{s}(X)$ & BFV secret-key polynomial. \\

$\hat{\epsilon}$ & BFV noise polynomial. 
&
$\chi_e$ & BFV error distribution. \\

$\hat{m}$ & Plaintext polynomial. 
&
$\llbracket \hat{m}\rrbracket$ & BFV encryption of $\hat{m}$. \\

$(\hat{c_0},\hat{c_1})$ & Fresh BFV ciphertext components. 
&
$(\hat{c}'_0,\hat{c}'_1)$ & Ciphertext components after truncation. \\

$\langle x\rangle$ & Additive sharing of $x$ over $\mathbb{Z}_{2^\ell}$.
&
$\langle x\rangle_C,\langle x\rangle_S$ & Client and server shares. \\

$r_{o,\xi}$ & Server random mask/share for output $(o,\xi)$.
&
$\mathbf{r}$ & Vector of random masks. \\

\midrule
\multicolumn{4}{@{}l}{\textbf{Convolution and SPA-CONV}} \\
\midrule
$\mathcal{T}$ & Input activation tensor.
&
$\mathcal{W}$ & Convolution weight tensor. \\

$H,W$ & Input height and width.
&
$H',W'$ & Output height and width. \\

$K$ & Kernel size; kernel is $K\times K$.
&
$C_{\mathrm{in}},C_{\mathrm{out}}$ & Input/output channel counts. \\

$\Omega_K$ & Spatial offsets $\{-r,\ldots,r\}^2$.
&
$(u,v)$ & Spatial offset in $\Omega_K$. \\

$s_{\mathrm{stride}}$ & Convolution stride.
&
$p_{\mathrm{pad}}$ & Padding size. \\

$W_{\mathrm{pp}}$ & Padded row width.
&
$\delta(u,v)$ & Spatial coefficient offset $uW_{\mathrm{pad}}+v$. \\

$\Lambda$ & Input packing map.
&
$\Lambda(c,j)=(\iota,\sigma)$ & Source ciphertext index and coefficient offset. \\

$\mathrm{CT}_t$ & $t$-th encrypted input block.
&
$M$ & Number of encrypted input blocks. \\

$B$ & Output ciphertext blocks per output channel.
&
$j$ & Output block index. \\

$\nu$ & Split factor.
&
$o$ & Output channel index. \\

$\mathcal{C}(o)$ & Contributing input-channel set for output $o$.
&
$T_o$ & Number of shifted scalar terms, $|\mathcal{C}(o)|K^2$. \\

$A_{o,j}$ & Raw encrypted output block.
&
$Y_{o,j}$ & Truncated/reduced encrypted output block. \\

$\mathrm{Trunc}_\tau(\cdot)$ & Ciphertext-side $\tau$-bit truncation.
&
$\mathrm{MaskReduce}(\cdot)$ & Power-of-two modular reduction by masking. \\

\midrule
\multicolumn{4}{@{}l}{\textbf{Grouped export and reconstruction}} \\
\midrule
$\Phi$ & Output extraction map.
&
$\Phi(o,\xi)=(\gamma,\kappa)$ & Source ciphertext and coefficient index. \\

$\xi$ & Logical output spatial position.
&
$\kappa$ & Coefficient index inside an RLWE ciphertext. \\

$\gamma$ & Source RLWE ciphertext group index.
&
$\mathcal{P}_\gamma$ & Logical outputs extracted from source $\gamma$. \\

$G_\gamma$ & Grouped export record for source $\gamma$.
&
$\mathcal{G}$ & Set/list of grouped export records. \\

$G$ & Number of non-empty export groups.
&
$m_g$ & Number of exported coefficients in group $g$. \\

$\beta_g$ & Transport bitwidth for group $g$.
&
$E$ & Total exported values, $\sum_g m_g=C_{\mathrm{out}}H'W'$. \\

$\hat{a}_\gamma(X)$ & Shared RLWE linear component for group $\gamma$.
&
$\hat{b}_\gamma(X)$ & RLWE constant component for group $\gamma$. \\

$\mathbf{b}_\gamma$ & Vector of exported constant coefficients.
&
$\boldsymbol{\kappa}_\gamma$ & Vector of coefficient indices. \\

$\mathbf{ch}_\gamma$ & Vector of output-channel metadata.
&
$\mathbf{off}_\gamma$ & Vector of output-position metadata. \\

$\hat{d}_\gamma(X)$ & Client-side product $a_\gamma(X)\mathsf{sk}(X)$.
&
$u_i$ & Reconstructed masked coefficient before decoding. \\

$z_i$ & Decoded client share for one exported value.
&
$\mathbf{z}$ & Vector of decoded client shares. \\

\bottomrule
\end{longtable}
\endgroup



\section{Proof of Theorem~\ref{thm:precise_trunc}}
\label{app:proof}

\precisetruncthm*
\begin{lemma}[Unsigned shifts over power-of-two residues]
\label{lem:unsigned_shift}
Let $q=2^Q$ and $q_f=2^{Q-f}$ with $0<f\le Q$.
For any integer $\tilde{c}\in[-2^{Q-1},2^{Q-1})$, let
$c=\tilde{c}\bmod q\in[0,q)$ be its canonical unsigned residue.
Then
\[
    \left\lfloor \frac{c}{2^f} \right\rfloor
    \equiv
    \left\lfloor \frac{\tilde{c}}{2^f} \right\rfloor
    \pmod{q_f}.
\]
Thus, a logical right shift of the canonical unsigned residue modulo $2^Q$
represents the arithmetic right shift of the signed centered representative
modulo $2^{Q-f}$.
\end{lemma}

\begin{proof}[Proof of Lemma~\ref{lem:unsigned_shift}]
If $\tilde{c}\ge 0$, then $c=\tilde{c}$ and the claim is immediate.
If $\tilde{c}<0$, then $c=2^Q+\tilde{c}$. Since $2^f$ divides $2^Q$,
\[
    \left\lfloor \frac{c}{2^f} \right\rfloor
    =
    \left\lfloor \frac{2^Q+\tilde{c}}{2^f} \right\rfloor
    =
    2^{Q-f}
    +
    \left\lfloor \frac{\tilde{c}}{2^f} \right\rfloor.
\]
Therefore,
\[
    \left\lfloor \frac{c}{2^f} \right\rfloor
    \equiv
    \left\lfloor \frac{\tilde{c}}{2^f} \right\rfloor
    \pmod{2^{Q-f}},
\]
which proves the claim.
\end{proof}

\begin{proof}[Proof of Theorem~\ref{thm:precise_trunc}]
    By Lemma~\ref{lem:unsigned_shift}, the logical right shifts used in
    Algorithm~\ref{alg:trunc} on canonical unsigned residues are equivalent,
    modulo $q_f=2^{Q-f}$, to arithmetic right shifts of the corresponding signed
    centered representatives. Hence, below we analyze $\hat{c}'_0$ and $\hat{c}'_1$
    as coefficient-wise arithmetic shifts of $\hat{c}_0$ and $\hat{c}_1$ in the
    reduced ring $R_{q_f}$.

    Since $(\hat{c}_0, \hat{c}_1) \in R_q^2$ encrypts $\hat{m}_0 \in R_p$, we have $[\hat{c}_0 + \hat{c}_1\hat{s}]_q = \Delta \hat{m}_0 + \hat{\epsilon}$ for some noise polynomial $\hat{\epsilon}$. Hence, for some integer-coefficient polynomial $\hat{r}$,
    \[
        \frac{\hat{c}_0 + \hat{c}_1\hat{s}}{2^f} = \frac{\Delta \hat{m}_0}{2^f} + \frac{\hat{\epsilon}}{2^f} + \frac{q}{2^f}\hat{r}
    \]
    Using $\lfloor x \rfloor \leq x < \lfloor x \rfloor+1$ for $\forall x \in \mathbb{R}$, and $\hat{s}[i] \in \{-1, 0, 1\}$ for $\forall i \in [N]$, 
    \[
       (\frac{\hat{c}_0 + \hat{c}_1\hat{s}}{2^f})[i] - N - 1 < (\hat{c}'_0 + \hat{c}'_1\hat{s})[i] < (\frac{\hat{c}_0 + \hat{c}_1\hat{s}}{2^f})[i] + N, \quad \forall i \in [N]
    \]
    \[
        \Rightarrow \; (\frac{\Delta \hat{m}_0}{2^f} + \frac{\hat{\epsilon}}{2^f} + \frac{q}{2^f}\hat{r})[i] - N - 1 < (\hat{c}'_0 + \hat{c}'_1\hat{s})[i] < (\frac{\Delta \hat{m}_0}{2^f} + \frac{\hat{\epsilon}}{2^f} + \frac{q}{2^f}\hat{r})[i]+N
    \]
    \[
        \Rightarrow \; (\hat{m}_0 + \frac{\hat{\epsilon}}{\Delta})[i] - \frac{(N+1)2^f}{\Delta} < (\frac{2^f}{\Delta}[\hat{c}'_0 + \hat{c}'_1\hat{s}]_{q_f})[i] < (\hat{m}_0 + \frac{\hat{\epsilon}}{\Delta})[i]+\frac{2^f}{\Delta}N
    \]
    As long as $||\hat{\epsilon}||_\infty < \Delta/2 - (N+1)2^f$, we get $\lfloor \frac{2^f}{\Delta}[\hat{c}'_0 + \hat{c}'_1\hat{s}]_{q_f} \rceil = \hat{m}_0$. Therefore, Algorithm~\ref{alg:dec} returns the plaintext $\hat{m}_1$, where $\hat{m}_1[i] = (\hat{m}_0[i] \gg f)$.
\end{proof}

\section{Noise-Budget Analysis for Maximum-Accumulation Layers}
\label{app:noise_budget}

To assess whether Jaguar's homomorphic convolution remains within the
ciphertext-noise budget, we analyze the layers with the largest
scalar--ciphertext accumulation in each evaluated network. For a fixed output
ciphertext coefficient, these layers maximize the number of accumulated
weighted ciphertext terms and therefore represent the most adverse candidates
for noise growth in the SPA-CONV path.

The maximum accumulation counts are $4608$ for ResNet-18 (layers 16, 18, and
19), $4608$ for ResNet-50 (layers 44, 48, and 51), and $960$ for MobileNetV2
(layers 44, 47, and 50). For each such layer, we reconstruct the BN-fused
convolution weights exactly as in the Jaguar implementation. Let
$w^{\mathrm{centered}}_{o,j}$ denote the centered convolution weight and
$\mathrm{BNMul}_o$ the BN scaling factor for output channel $o$. The fused
integer weight is
\[
w^{\mathrm{fused}}_{o,j}
=
\mathrm{Round}\!\left(
\frac{
w^{\mathrm{centered}}_{o,j}\cdot \mathrm{BNMul}_o
}{2^{s_{\mathrm{BN}}}}
\right),
\]
where $s_{\mathrm{BN}}$ is the BN-fusion scale parameter used by the runtime.

\paragraph{Gaussian noise propagation.}
The noise polynomial coefficients of a fresh symmetrically-encrypted BFV ciphertext are independent, zero-mean
Gaussian random variables,
\[
n_{\mathrm{in}}[j]\sim \mathcal{N}(0,\sigma_{\mathrm{in}}^2),
\qquad
\sigma_{\mathrm{in}}=3.2.
\]
For output channel $o$, the accumulated pre-truncation output noise is
\[
n_{\mathrm{out}}[o]
=
\sum_j
w^{\mathrm{fused}}_{o,j}\, n_{\mathrm{in}}[j].
\]
Under the independence assumption,
\[
\mathrm{Var}[n_{\mathrm{out}}[o]]
=
\sigma_{\mathrm{in}}^2
\sum_j
\left(w^{\mathrm{fused}}_{o,j}\right)^2,
\qquad
\sigma_{\mathrm{out}}[o]
=
\sigma_{\mathrm{in}}
\sqrt{
\sum_j
\left(w^{\mathrm{fused}}_{o,j}\right)^2
}.
\]
Thus, the noise amplification of each output channel is governed by the
$\ell_2$ norm of its BN-fused weight vector, not by the accumulation count
alone.

\paragraph{Consistent noise scale and threshold.}
The comparison must be performed at a consistent scale. Jaguar uses
$(Q,P,f)=(54,30,9)$, so
\[
\Delta=2^{Q-P}=2^{24}.
\]
The pre-truncation BFV rounding margin is
\[
B_{\mathrm{dec}}
=
\Delta/2
=
2^{Q-P-1}
=
2^{23}.
\]
The sufficient condition in Theorem~\ref{thm:precise_trunc} is slightly
stronger because it also accounts for the component-wise shift error:
\[
B_{\mathrm{trunc}}
=
\Delta/2-(N+1)2^f
=
2^{23}-2049\cdot2^9
=
7{,}339{,}520.
\]
Equivalently, after the $f$-bit shift, the corresponding theorem-level
threshold is
\[
\frac{B_{\mathrm{trunc}}}{2^f}
=
2^{Q-P-f-1}-(N+1)
=
14{,}335.
\]
Since the noise standard deviation is reduced by the same factor $2^f$ after
truncation, the effective margin is unchanged as long as both quantities are
compared at the same scale. In this appendix, we compare pre-truncation noise
estimates against the pre-truncation theorem bound $B_{\mathrm{trunc}}$.

\begin{figure}[t]
    \centering
    \includegraphics[width=0.85\linewidth]{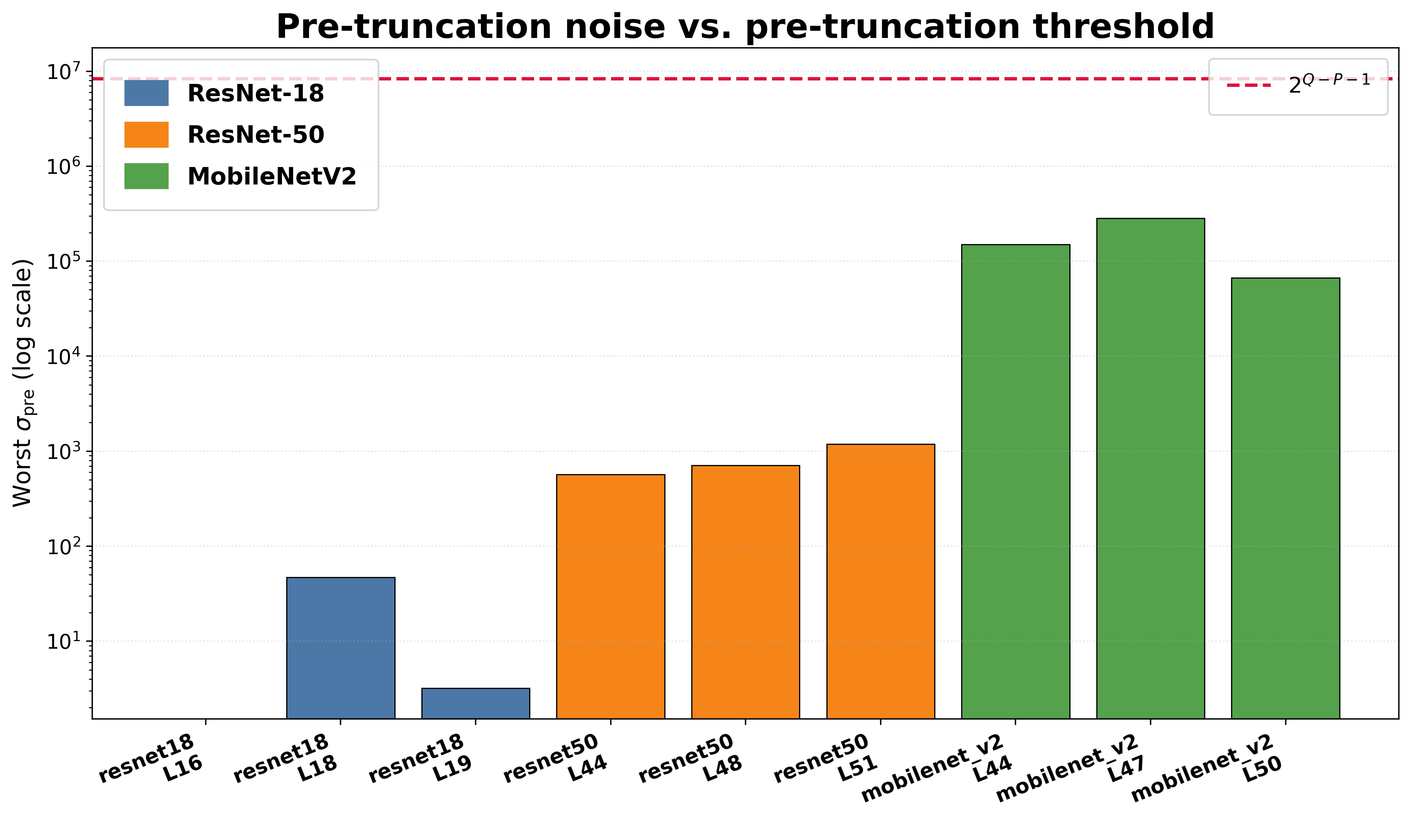}
    \caption{
    Worst-case pre-truncation output-noise standard deviation at the maximum
    scalar--ciphertext accumulation layers of Jaguar. Each bar corresponds to
    the output channel with the largest estimated $\sigma_{\mathrm{out}}$ in
    that layer. The dashed line marks the pre-truncation BFV rounding margin
    $B_{\mathrm{dec}}=2^{Q-P-1}=2^{23}$; the theorem-level bound is
    $B_{\mathrm{trunc}}=B_{\mathrm{dec}}-(N+1)2^f$. All estimated
    worst-channel noise levels remain far below both boundaries.
    }
    \label{fig:jaguar_noise_overall_pre}
\end{figure}

\begin{figure}[t]
    \centering
    \begin{subfigure}[t]{0.32\linewidth}
        \centering
        \includegraphics[width=\linewidth]{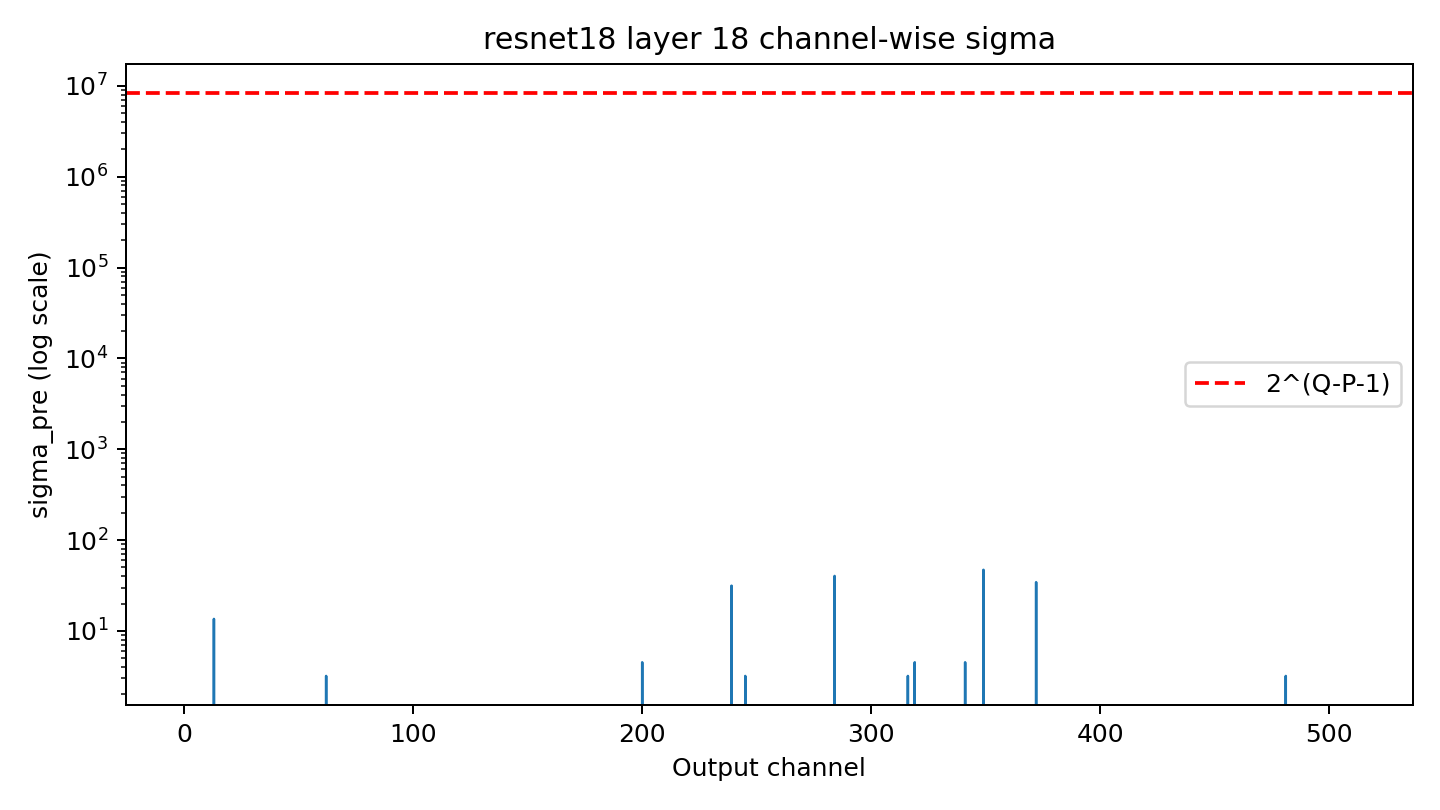}
        \caption{ResNet-18, layer 18}
        \label{fig:jaguar_sigma_resnet18_l18_pre}
    \end{subfigure}
    \hfill
    \begin{subfigure}[t]{0.32\linewidth}
        \centering
        \includegraphics[width=\linewidth]{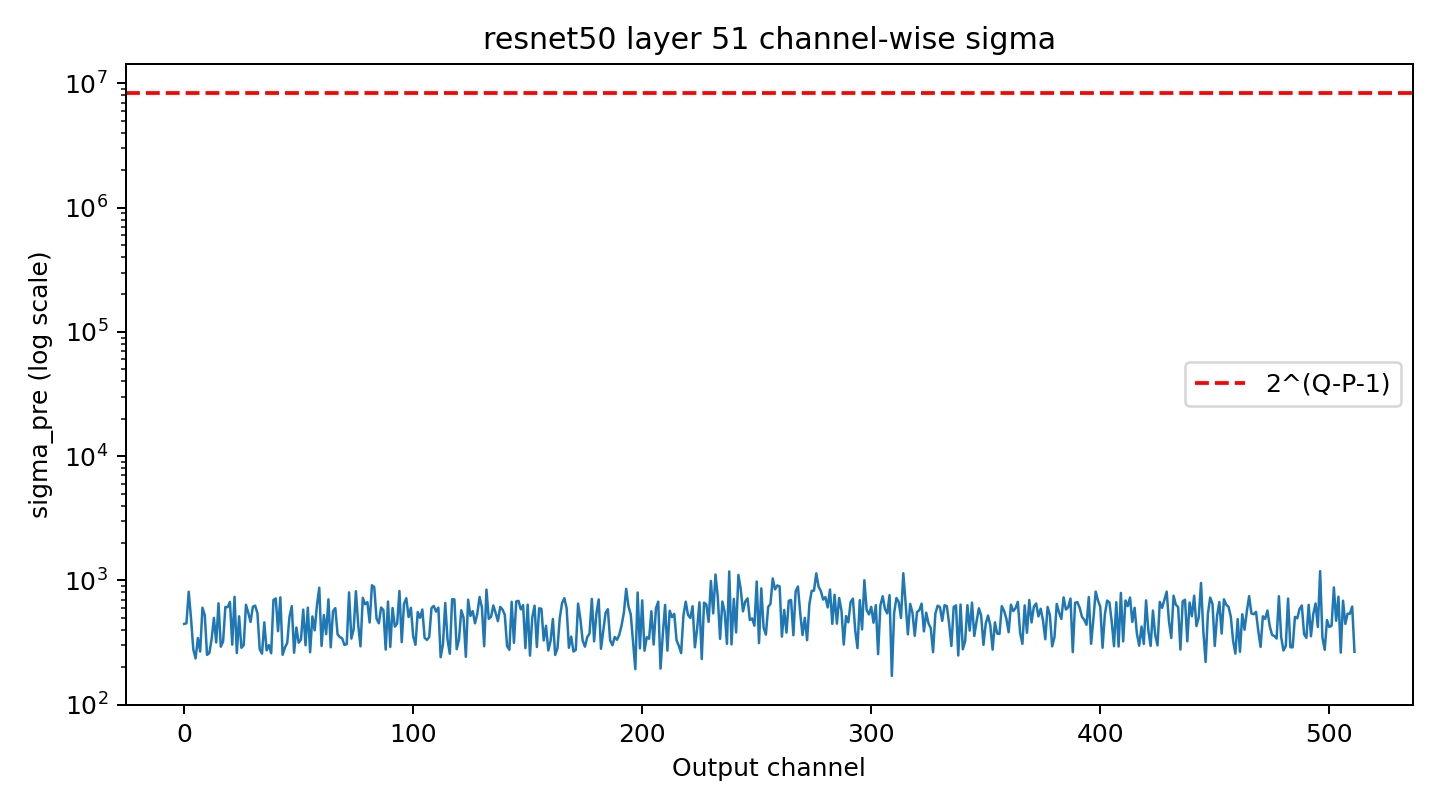}
        \caption{ResNet-50, layer 51}
        \label{fig:jaguar_sigma_resnet50_l51_pre}
    \end{subfigure}
    \hfill
    \begin{subfigure}[t]{0.32\linewidth}
        \centering
        \includegraphics[width=\linewidth]{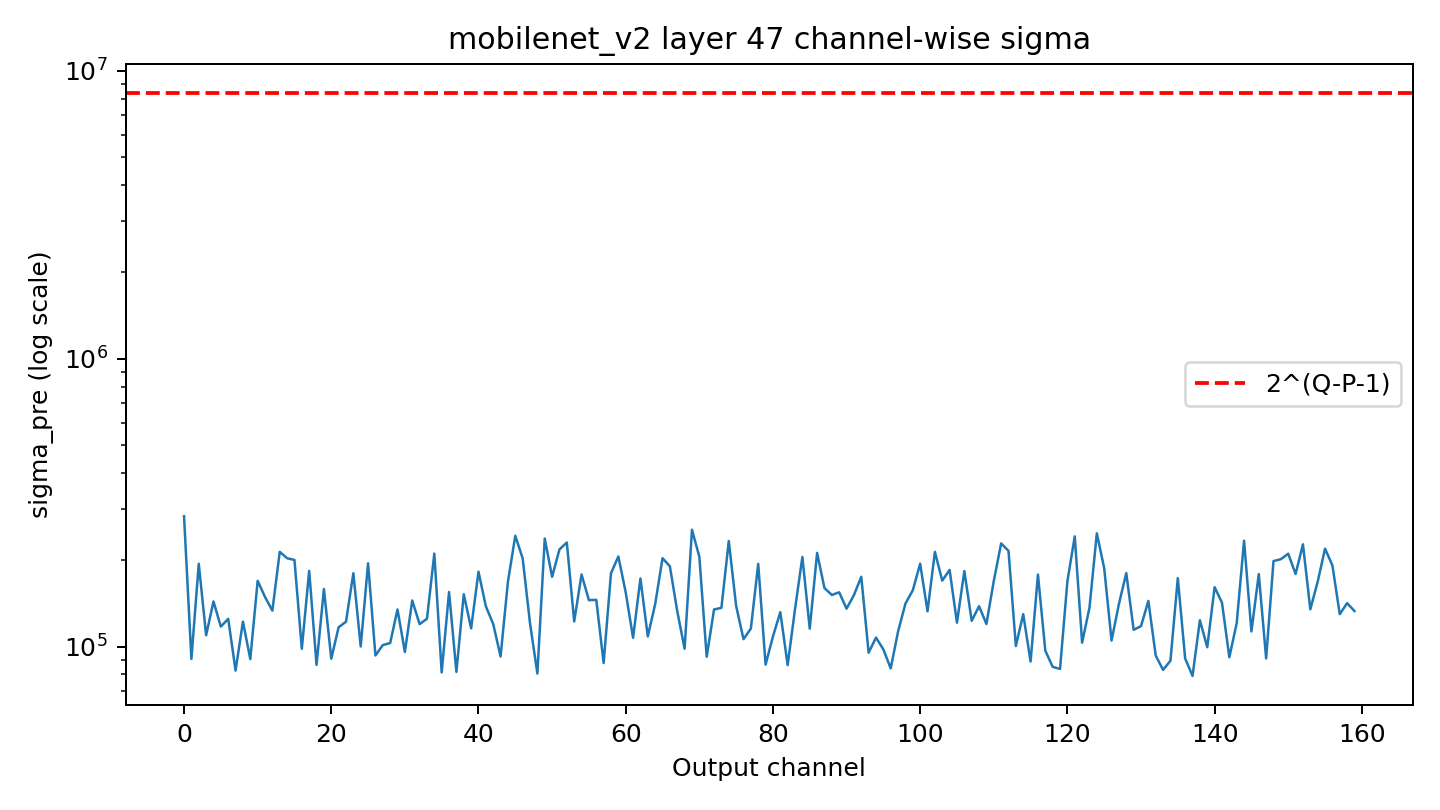}
        \caption{MobileNetV2, layer 47}
        \label{fig:jaguar_sigma_mobilenet_l47_pre}
    \end{subfigure}
    \caption{
    Channel-wise pre-truncation output-noise standard deviation for
    representative maximum-accumulation layers. The variation across output
    channels reflects the $\ell_2$ norms of the corresponding BN-fused weight
    vectors.
    }
    \label{fig:jaguar_channelwise_sigma_pre}
\end{figure}

\paragraph{Results.}
The largest estimated pre-truncation standard deviations among the analyzed
maximum-accumulation layers are $47.14$ for ResNet-18, $1187.05$ for
ResNet-50, and $2.84\times10^5$ for MobileNetV2. Relative to
$B_{\mathrm{trunc}}=7{,}339{,}520$, these correspond to margins of
approximately $1.56\times10^5\sigma$, $6.18\times10^3\sigma$, and
$25.9\sigma$, respectively. The corresponding Gaussian tail probability
\[
\Pr\!\left(
\left|\mathcal{N}(0,\sigma_{\max}^2)\right|
>
B_{\mathrm{trunc}}
\right)
=
\mathrm{erfc}\!\left(
\frac{B_{\mathrm{trunc}}}{\sqrt{2}\sigma_{\max}}
\right)
\]
is numerically negligible; for the worst MobileNetV2 layer, it is approximately
$9.3\times10^{-148}$.

MobileNetV2 exhibits the largest estimated noise growth, especially at layer
47. This is not caused by missing BN fusion. Rather, it results from a
$960$-term pointwise ($1\times1$) accumulation combined with relatively large
BN scaling coefficients, which increase the $\ell_2$ norms of the BN-fused
weights. In contrast, the ResNet layers accumulate more terms
($512\times3\times3=4608$), but their BN-fused weight norms are much smaller,
so their propagated noise remains substantially lower.

This analysis 
provides a conservative sanity
check for the most accumulation-heavy homomorphic convolution layers used in
Jaguar: under the Gaussian propagation model, the BN-fused weight norms of the
analyzed worst-case layers remain far from the theorem-level truncation boundary.

\section{Security and Parameter Selection}
\label{app:security_guarantee}

Jaguar instantiates BFV over the negacyclic ring
$R_q=\mathbb{Z}_q[X]/(X^N+1)$ with a power-of-two ciphertext modulus
$q=2^Q$ and a power-of-two plaintext modulus $p=2^P$. The parameter set used in
our main evaluation is
\[
    (N,Q,P)=(2048,54,30),
    \qquad q=2^{54}, \quad p=2^{30}.
\]
This appendix explains how these parameters are selected.

\paragraph{Security target.}
The security of BFV is determined by the concrete hardness of the underlying
RLWE instance, which depends on the ring dimension, ciphertext modulus, secret
distribution, and error distribution, rather than on whether the ciphertext
modulus is prime. We follow the Homomorphic Encryption Standard and evaluate the
parameter set using the standard attacks against RLWE~\cite{homencstandard}.
For the small-secret distribution used in Jaguar, the HE standard table with
classical BKZ.sieve cost model gives the following estimates for
$N=2048$ and $\log_2 q=54$:
\[
    \lambda_{\mathrm{uSVP}}=129.7,\qquad
    \lambda_{\mathrm{dec}}=144.4,\qquad
    \lambda_{\mathrm{dual}}=134.2.
\]
We take the minimum over these attacks, yielding a classical security estimate
of $129.7$ bits. Thus, $(N,Q)=(2048,54)$ satisfies the classical 128-bit
security target. This choice also matches the 128-bit coefficient-modulus
budget commonly used by BFV implementations for $N=2048$~\cite{sealmanual}.

\paragraph{Why $N=2048$.}
The ring dimension $N$ directly affects ciphertext size and the cost of
coefficient-wise arithmetic. Since Jaguar's convolution kernel is built from
scalar--polynomial products, coefficient shifts, additions, and masks, its
linear-layer cost scales linearly with $N$ for each ciphertext block. We
therefore choose the smallest ring dimension that supports the required
plaintext space and satisfies the classical 128-bit security target. Under the
HE standard table, $N=2048$ permits a 54-bit ciphertext modulus at 128-bit
classical security, making it the most efficient choice for our target
parameter range.

\paragraph{Why $P=30$.}
The plaintext modulus determines the integer range used for fixed-point
inference and the ring used by the subsequent 2PC protocols. We use
$p=2^{30}$ so that HE outputs are naturally represented over the same
power-of-two ring as the additive shares used by the non-linear protocols.
The choice $P=30$ also provides sufficient dynamic range for the fixed-point
configuration used in our ImageNet models. In Appendix~\ref{app:param_selection}, we
validate this choice by comparing fixed-point inference against the FP32
torchvision baseline.

\paragraph{Why $Q=54$.}
Given $N=2048$, increasing $Q=\log_2 q$ increases the available noise budget but
reduces the concrete RLWE security level. We choose $Q=54$ because it is the
largest ciphertext-modulus bit-length allowed by the classical 128-bit security
table for the selected ring dimension. This maximizes the noise budget while
remaining within the target security level.

\paragraph{Compatibility with ciphertext-side truncation.}
For the selected parameters, the BFV scale is
\[
    \Delta=\frac{q}{p}=2^{Q-P}=2^{24}.
\]
Jaguar performs ciphertext-side truncation by $f$ bits, where $f=9$ is the
fixed-point fractional precision used in our evaluation. Since
\[
    2^f = 2^9 \mid \Delta = 2^{24},
\]
the divisibility condition required by Section~3.1 is satisfied. After
truncation, the effective scale remains $\Delta/2^f=2^{15}$, which is sufficient
for decoding target-precision fixed-point values.

The correctness condition from Section~3.1 requires the accumulated BFV noise
before truncation to remain below
\[
    \frac{\Delta}{2} - (N+1)2^f.
\]
With $N=2048$, $\Delta=2^{24}$, and $f=9$, this bound becomes
\[
    2^{23} - 2049\cdot 2^9 = 7{,}339{,}520.
\]
Jaguar's HE computation has multiplicative depth one per linear block before
the result is converted back to additive shares. Therefore, the noise growth is
limited to one plaintext--ciphertext linear evaluation, and we empirically
validate correctness by checking that the decrypted, truncated HE outputs match
the corresponding fixed-point reference outputs for all evaluated layers.

\paragraph{Classical versus post-quantum security.}
The parameter set above is selected for classical 128-bit security. Under the
post-quantum BKZ.qsieve table in the HE standard, a smaller modulus budget is
required for $N=2048$ to claim 128-bit post-quantum security. If post-quantum
128-bit security is required, Jaguar can either reduce $Q$ accordingly or
increase the ring dimension, e.g., to $N=4096$, while keeping the same
power-of-two arithmetic design. The algorithms in Sections~3 and~4 are unchanged
by this parameter adjustment, although the per-ciphertext arithmetic cost grows
with $N$.

\section{Parameter Selection for Fixed-Point Inference}
\label{app:param_selection}

\vspace{-2mm}
\begin{table}[h]
\centering
\caption{Agreement with the FP32 torchvision baseline from actual OpenCheetah executions on a sampled validation subset, using the original setting and the tested $(f,P)$ candidates.}
\vspace{2mm}
\label{tab:kell_selection}
\setlength{\tabcolsep}{5pt}
\renewcommand{\arraystretch}{1.15}
\begin{tabular}{lccccc}
\toprule
\multirow{2}{*}{Model}
& \multicolumn{5}{c}{OpenCheetah fixed-point setting $(f,P)$} \\
\cmidrule(lr){2-6}
& \makecell{Default\\$(12,37)$}
& $(9,25)$
& $(9,27)$
& $(9,29)$
& $(9,30)$ \\
\midrule
ResNet18 & 97.0\% & 0.0\% & 77.0\% & 85.0\% & 93.5\% \\
ResNet50 & 97.5\% & 0.0\% & 6.0\% & 8.0\% & 96.0\% \\
MobileNetV2 & 87.0\% & 0.0\% & 87.5\% & 92.0\% & 94.5\% \\
\bottomrule
\end{tabular}
\end{table}

We re-selected $(f,P)$ using a three-stage procedure after extending the OpenCheetah~\cite{opencheetah} evaluation beyond the originally provided models to additional architectures, including ResNet18 and MobileNetV2. 

In the first stage, we used plaintext fixed-point simulation to screen candidate settings by measuring their agreement with the FP32 torchvision baseline under ideal arithmetic. This stage captures deterministic fixed-point effects such as quantization error and range overflow, but it is not sufficient for the final parameter choice because the actual OpenCheetah pipeline introduces approximation in the nonlinear stage.

In the second stage, we used a protocol-aware plaintext simulator to account for the numerical behavior of OpenCheetah's nonlinear primitives without executing the cryptographic protocols.
On full 50,000 validation images, the protocol-aware simulation yields FP32 agreement of 85.918\%/93.112\% for ResNet-18 and 9.632\%/96.008\% for ResNet-50 at $(f,P)=(9,29)/(9,30)$, respectively.

In the third stage, given the high cost of end-to-end secure execution,
we validated the shortlisted candidates by running OpenCheetah on a
sampled subset of the validation set and comparing their outputs against
the FP32 baseline.

As shown in the table above, the original OpenCheetah reference setting $(f,P)=(12,37)$ already agrees well with the FP32 baseline on ResNet-18 and ResNet-50 (97.0\% and 97.5\%, respectively), but drops to 87.0\% on MobileNetV2.
Under $f=9$, the setting $(f,P)=(9,30)$ achieves FP32 agreement of 93.5\%, 96.0\%, and 94.5\% on ResNet-18, ResNet-50, and MobileNetV2, respectively.
For both ResNet-18 and ResNet-50, the sampled secure-execution agreement rates follow the same trends as the full-validation simulation in the second stage.
These results provide implementation-level confirmation of the simulation-based screening.

As an additional label-level sanity check, the selected setting $(f,P)=(9,30)$ achieves sampled top-1 accuracies of $74\%$, $84.5\%$, and $77.0\%$ on ResNet18, ResNet50, and MobileNetV2, respectively, which are comparable to the original OpenCheetah setting ($75.5\%$, $82.0\%$, and $70.0\%$).
Based on these results, we use $(f,P)=(9,30)$ in all subsequent experiments.

\section{Additional Experimental Results}
\begin{table*}[h]
\centering
\small
\setlength{\tabcolsep}{4pt}
\renewcommand{\arraystretch}{1.08}
\caption{
Case-level compute-only convolution microbenchmark with AVX disabled and a single thread.
Latency is averaged over 10 runs.
Speedup is computed as Cheetah over Jaguar.
}
\label{tab:conv-case-avxoff}
\begin{tabular}{l l c c c c}
\toprule
Conv Type & Layout & Dimension $(W, C_{\mathrm{in}}, C_{\mathrm{out}}, K)$ & Cheetah (ms) & Jaguar (ms) & Speedup \\
\midrule

\multirow{14}{*}{Dense 3$\times$3}
& \multirow{10}{*}{packed}
& $(32, 128, 128, 3)$ & 2631.60 & 272.95 & 9.64$\times$ \\
& & $(32, 16, 16, 3)$ & 45.21 & 4.82 & 9.38$\times$ \\
& & $(28, 128, 128, 3)$ & 1975.56 & 266.77 & 7.41$\times$ \\
& & $(28, 512, 512, 3)$ & 30956.48 & 4303.07 & 7.19$\times$ \\
& & $(16, 64, 64, 3)$ & 180.03 & 69.03 & 2.61$\times$ \\
& & $(16, 256, 256, 3)$ & 2699.46 & 1099.17 & 2.46$\times$ \\
& & $(14, 256, 256, 3)$ & 2007.58 & 1132.27 & 1.77$\times$ \\
& & $(8, 128, 128, 3)$ & 240.01 & 273.78 & 0.88$\times$ \\
& & $(7, 512, 512, 3)$ & 2717.50 & 4376.67 & 0.62$\times$ \\
& & $(4, 512, 512, 3)$ & 1241.24 & 4347.74 & 0.29$\times$ \\
\cmidrule(lr){2-6}
& \multirow{4}{*}{split}
& $(64, 64, 64, 3)$ & 7715.16 & 204.37 & 37.75$\times$ \\
& & $(56, 64, 64, 3)$ & 1958.06 & 134.66 & 14.54$\times$ \\
& & $(56, 256, 256, 3)$ & 30964.59 & 2165.62 & 14.30$\times$ \\
& & $(224, 64, 64, 3)$ & 30962.18 & 2417.93 & 12.81$\times$ \\

\midrule
\multirow{13}{*}{Pointwise 1$\times$1}
& \multirow{9}{*}{packed}
& $(28, 128, 512, 1)$ & 6030.81 & 129.04 & 46.73$\times$ \\
& & $(28, 512, 128, 1)$ & 5896.45 & 126.95 & 46.45$\times$ \\
& & $(28, 32, 192, 1)$ & 614.23 & 14.55 & 42.20$\times$ \\
& & $(14, 1024, 256, 1)$ & 6060.75 & 501.12 & 12.09$\times$ \\
& & $(14, 256, 1024, 1)$ & 6088.65 & 522.41 & 11.65$\times$ \\
& & $(14, 96, 576, 1)$ & 1344.09 & 121.90 & 11.03$\times$ \\
& & $(7, 512, 2048, 1)$ & 6899.98 & 2025.97 & 3.41$\times$ \\
& & $(7, 2048, 512, 1)$ & 6083.87 & 1983.25 & 3.07$\times$ \\
& & $(7, 320, 1280, 1)$ & 2494.59 & 821.44 & 3.04$\times$ \\
\cmidrule(lr){2-6}
& \multirow{4}{*}{split}
& $(56, 64, 256, 1)$ & 937.30 & 72.18 & 12.99$\times$ \\
& & $(56, 256, 64, 1)$ & 939.37 & 91.27 & 10.29$\times$ \\
& & $(56, 24, 144, 1)$ & 206.50 & 20.41 & 10.12$\times$ \\
& & $(112, 32, 16, 1)$ & 102.24 & 10.31 & 9.91$\times$ \\

\midrule
\multirow{9}{*}{Depthwise 3$\times$3}
& \multirow{6}{*}{packed}
& $(14, 96, 96, 3)$ & 50.29 & 3.86 & 13.04$\times$ \\
& & $(28, 32, 32, 3)$ & 17.90 & 1.42 & 12.57$\times$ \\
& & $(28, 192, 192, 3)$ & 98.95 & 8.09 & 12.23$\times$ \\
& & $(7, 320, 320, 3)$ & 164.85 & 13.54 & 12.18$\times$ \\
& & $(7, 160, 160, 3)$ & 87.60 & 8.23 & 10.65$\times$ \\
& & $(14, 576, 576, 3)$ & 294.73 & 30.57 & 9.64$\times$ \\
\cmidrule(lr){2-6}
& \multirow{3}{*}{split}
& $(56, 24, 24, 3)$ & 14.47 & 2.18 & 6.65$\times$ \\
& & $(56, 144, 144, 3)$ & 78.72 & 13.89 & 5.67$\times$ \\
& & $(112, 32, 32, 3)$ & 70.89 & 16.52 & 4.29$\times$ \\

\bottomrule
\end{tabular}
\end{table*}

\subsection{Full Convolution Microbenchmark Results}
\label{app:full_conv_microbench}

Table~\ref{tab:conv-case-avxoff} provides the full per-case results behind the grouped summary in Table~\ref{tab:conv_microbench_avxoff}. 
We use the same compute-only setup as in the main microbenchmark: each experiment measures only the server-side convolution computation, excluding key generation, encryption, decryption, communication, and non-linear protocols. 
All runs are executed with a single thread and AVX disabled. 
For each shape, we repeat the measurement 10 times and report the average latency. 
Speedup is computed from the paired average latency of Cheetah over Jaguar.

We group cases by convolution type and ciphertext layout. 
Since all spatial dimensions are square, each shape is written as $(W,C_{\mathrm{in}},C_{\mathrm{out}},K)$, where $W$ is the spatial width, $C_{\mathrm{in}}$ and $C_{\mathrm{out}}$ are the input and output channel counts, and $K$ is the kernel size. 
For depthwise convolutions, $C_{\mathrm{out}}=C_{\mathrm{in}}$.

\subsection{AVX-512 Implementation}
\label{app:avx}
\noindent
\begin{minipage}[t]{0.43\columnwidth}
\vspace{0pt}
Jaguar's power-of-two modulus is well-suited to AVX-512; reduction modulo $2^Q$ is a vector bitwise AND and ciphertext-side truncation is a vector arithmetic right shift, so the dominant 64-bit integer kernels in \textsc{SPA-Conv} (scalar--polynomial multiplication, shifted accumulation, server-share addition, truncation, reduction) map directly to Intel vector instructions~\cite{intel_avx512_overview,intel_intrinsics_guide}. 

\end{minipage}
\hfill
\begin{minipage}[t]{0.54\columnwidth}
\vspace{0pt}
\centering
\scriptsize
\setlength{\tabcolsep}{2.4pt}
\renewcommand{\arraystretch}{1.08}

\captionof{table}{
End-to-end latency of Jaguar with and without AVX-512 under LAN and WAN settings.
Numbers in parentheses are latency reduction over AVX OFF.
}
\label{tab:jaguar-avx-latency}
\begin{tabular}{c cc cc}
\toprule
\multirow{2}{*}[-0.6ex]{Model}
& \multicolumn{2}{c}{LAN Latency (s)}
& \multicolumn{2}{c}{WAN Latency (s)} \\
\cmidrule(){2-3} \cmidrule(l){4-5}
& AVX ON & AVX OFF & AVX ON & AVX OFF \\
\midrule
ResNet18 & \textbf{7.90} (8.6\%) & 8.64 & \textbf{15.10} (5.2\%) & 15.94 \\
ResNet50 & \textbf{25.55} (21.3\%) & 32.44 & \textbf{44.91} (12.8\%) & 51.50 \\
MobileNetV2 & \textbf{16.83} (8.6\%) & 18.41 & \textbf{31.81} (1.9\%) & 32.42 \\
\bottomrule
\vspace{0.3mm}
\end{tabular}

\end{minipage}
AVX consistently lowers end-to-end latency without changing communication (Table~\ref{tab:jaguar-avx-latency}); largest gain (21.3\,\% on ResNet-50 LAN) appears where large convolution layers expose the most coefficient-wise work. MobileNetV2's smaller gain reflects more small layers and non-vectorized protocol overheads.

\begin{table*}[h]
\centering
\footnotesize
\setlength{\tabcolsep}{7pt}
\renewcommand{\arraystretch}{0.98}
\setlength{\aboverulesep}{0.25ex}
\setlength{\belowrulesep}{0.25ex}
\caption{
End-to-end latency and communication with AVX enabled.
Speedup and communication reduction are computed relative to Cheetah.
}
\label{tab:e2e_avxon}
\vspace{0mm}
\begin{tabular}{@{}cccccccc@{}}
\toprule
\multirow{2}{*}[-0.5ex]{Model} &
\multirow{2}{*}[-0.5ex]{Protocol} &
\multicolumn{2}{c}{LAN} &
\multicolumn{2}{c}{WAN} &
\multicolumn{2}{c}{Communication} \\
\cmidrule(lr){3-4}
\cmidrule(lr){5-6}
\cmidrule(l){7-8}
& & Latency (s) & Speedup & Latency (s) & Speedup & Comm. (MiB) & Reduction \\
\midrule

\multirow{3}{*}{ResNet18}
& Cheetah & 22.22 & 1.00$\times$ & 29.27 & 1.00$\times$ & 366.94 & 1.00$\times$ \\
& Rhombus & 21.96 & 1.01$\times$ & 29.24 & 1.00$\times$ & 344.01 & 1.07$\times$ \\
& \textbf{Jaguar} & \textbf{7.90} & \textbf{2.81}$\times$ & \textbf{15.10} & \textbf{1.94}$\times$ & \textbf{315.02} & \textbf{1.16}$\times$ \\
\midrule

\multirow{3}{*}{ResNet50}
& Cheetah & 77.89 & 1.00$\times$ & 103.72 & 1.00$\times$ & 1335.64 & 1.00$\times$ \\
& Rhombus & 66.27 & 1.18$\times$ & 80.83 & 1.28$\times$ & \textbf{750.22} & \textbf{1.78}$\times$ \\
& \textbf{Jaguar} & \textbf{25.55} & \textbf{3.05}$\times$ & \textbf{44.91} & \textbf{2.31}$\times$ & 962.0 & 1.39$\times$ \\
\midrule

\multirow{3}{*}{MobileNetV2}
& Cheetah & 41.02 & 1.00$\times$ & 61.20 & 1.00$\times$ & 1042.63 & 1.00$\times$ \\
& Rhombus & 46.75 & 0.88$\times$ & 68.29 & 0.90$\times$ & 836.81 & 1.25$\times$ \\
& \textbf{Jaguar} & \textbf{16.83} & \textbf{2.44}$\times$ & \textbf{31.81} & \textbf{1.92}$\times$ & \textbf{592.5} & \textbf{1.76}$\times$ \\
\bottomrule
\end{tabular}
\end{table*}

Table~\ref{tab:e2e_avxon} reports the end-to-end comparison under the AVX-enabled configuration.
This table complements Table~\ref{tab:jaguar-avx-latency}: while Table~\ref{tab:jaguar-avx-latency}
isolates the effect of AVX-512 within Jaguar, Table~\ref{tab:e2e_avxon} shows the performance of
the optimized Jaguar implementation against prior protocols. With AVX enabled, Jaguar achieves
2.44--3.05$\times$ lower LAN latency and 1.92--2.31$\times$ lower WAN latency than Cheetah, while
also reducing communication by 1.16--1.76$\times$. The latency benefit is largest on ResNet-50, where
large convolution layers expose substantial coefficient-wise work for \textsc{SPA-Conv} and AVX-512
vectorization. The WAN speedups are smaller than the LAN speedups because communication and
round-trip latency, which are unaffected by AVX, account for a larger fraction of end-to-end time.

\section{Limitation and Future Work}
\label{app:limitation}

\noindent
\begin{minipage}[t]{0.52\linewidth}
\paragraph{FC/MatMul extension.}
Jaguar targets the convolution path, which accounts for >\,98\,\% of linear-layer latency in our evaluated CNNs.
The current CNN implementation retains Cheetah's FC backend for the final classifier. Static-weight MatMul can also use Jaguar's arithmetic, but its efficiency depends on matrix shape, packing, and sparsity.
Figure~\ref{fig:fc_density} shows a complementary opportunity: as weight density decreases, Jaguar-style accumulation cost drops linearly with the number of nonzero weights, while the Cheetah-style baseline remains flat, with a crossover around density $0.0295$. Sparse and pruned MatMul are therefore a natural setting for extending Jaguar's regime to FC layers.
\end{minipage}
\hfill
\begin{minipage}[t]{0.45\linewidth}
\vspace{0pt}
\centering
\vspace{0mm}
\includegraphics[width=0.8\linewidth]{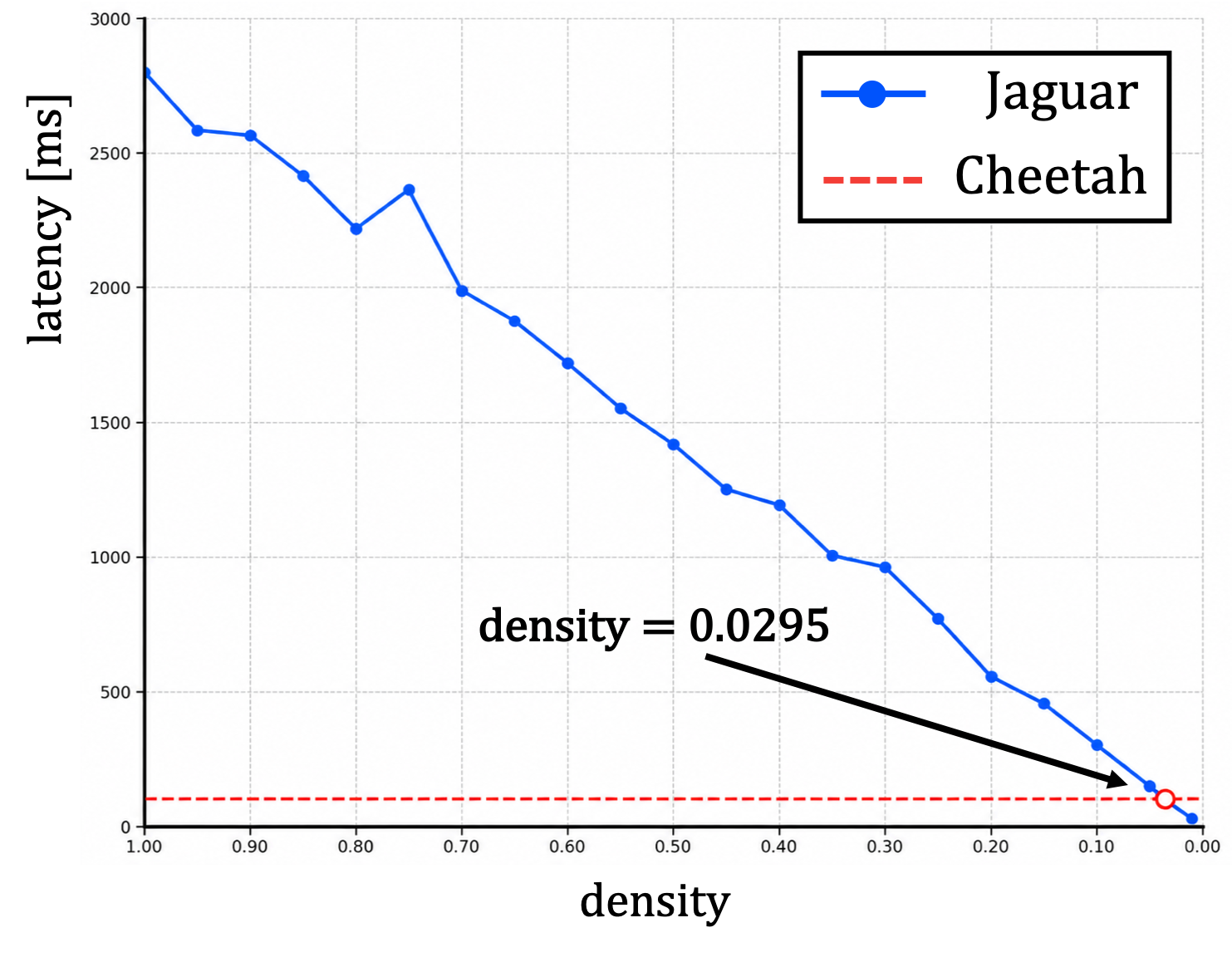}
\vspace{-2mm}
\captionof{figure}{FC/MatMul latency trend under different weight densities.}
\label{fig:fc_density}
\end{minipage}

\begin{table}[t]
\centering
\caption{Static-weight operator microbenchmarks in attention architectures.
All latencies are in milliseconds, and speedup is Cheetah/Jaguar.
(a) Dense MaxViT-T projections: medians of nine runs with matched precision and HE-linear timing boundaries; $n$ is the graph occurrence count.
(b) Sparse BERT-base MatMul: Jaguar latency at each weight density, compared with a density-independent Cheetah HE-linear baseline; Jaguar reports the best of five runs and Cheetah one run per shape.
Neither part measures end-to-end inference.}
\label{tab:attention_extension}
\small
\setlength{\tabcolsep}{5pt}
\textbf{(a) Dense MaxViT-T projections}\\[3pt]
\begin{tabular}{rrlrrrr}
\toprule
$T$ & $C_{\mathrm{in}}\!\to\!C_{\mathrm{out}}$ & Projection & $n$ & Jaguar & Cheetah & Speedup \\
\midrule
3136 & $64\to64$   & Output & 4  & 11.05   & 71.64   & $6.48\times$ \\
3136 & $64\to192$  & QKV    & 4  & 34.80   & 234.16  & $6.73\times$ \\
784  & $128\to128$ & Output & 4  & 26.88   & 292.73  & $10.89\times$ \\
784  & $128\to384$ & QKV    & 4  & 81.78   & 902.20  & $11.03\times$ \\
196  & $256\to256$ & Output & 10 & 93.82   & 294.16  & $3.14\times$ \\
196  & $256\to768$ & QKV    & 10 & 285.62  & 929.73  & $3.26\times$ \\
49   & $512\to512$ & Output & 4  & 379.54  & 320.05  & $0.84\times$ \\
49   & $512\to1536$& QKV    & 4  & 1144.94 & 1004.97 & $0.88\times$ \\
\midrule
\multicolumn{3}{l}{Occurrence-weighted sum} & 44 & 10510.35 & 23541.91 & $2.24\times$ \\
\bottomrule
\end{tabular}

\vspace{7pt}
\textbf{(b) Sparse BERT-base static-weight MatMul}\\[3pt]
\begin{tabular}{llrrrrr}
\toprule
& & \multicolumn{3}{c}{Jaguar at weight density} & & \\
\cmidrule(lr){3-5}
Weight shape & Sparsity & 5\% & 3\% & 1\% & Cheetah & Speedup at 1\% \\
\midrule
$768\times768$  & Unstructured       & 20.52 & 12.65 & 5.11  & 12.41 & $2.43\times$ \\
$768\times768$  & $16\times16$ block & 19.10 & 11.78 & 4.22  & 12.41 & $2.94\times$ \\
$768\times3072$ & Unstructured       & 87.63 & 54.92 & 23.77 & 51.19 & $2.15\times$ \\
$768\times3072$ & $16\times16$ block & 81.18 & 49.54 & 20.33 & 51.19 & $2.52\times$ \\
\bottomrule
\end{tabular}
\end{table}

Table~\ref{tab:attention_extension}(b) evaluates sparse BERT-base output-projection and FFN-expansion shapes using Jaguar's $(N,Q,P,f)=(2048,54,30,9)$ configuration. Jaguar's timing includes sparse-metadata traversal, accumulation, and output construction but excludes metadata preprocessing; Cheetah measures its server-side HE linear loop. Across unstructured and $16\times16$ block sparsity, speedups are $0.93$--$1.05\times$ at 3\% density and $2.15$--$2.94\times$ at 1\%.

\paragraph{Mapping static projections to SPA-CONV.}
For token features $X\in\mathbb{Z}_{2^P}^{T\times d}$ and server-held static weights $W\in\mathbb{Z}_{2^P}^{d\times d'}$, a projection computes
\[
    Y_{t,o}=\sum_{c=1}^{d}X_{t,c}W_{c,o}.
\]
Treating tokens as spatial positions makes this precisely a $K=1$ convolution: the same channel-mixing weights are reused at every position.
Fused QKV projection uses $W=[W_Q\;W_K\;W_V]$; the attention output projection and MLP linear layers have the same static-weight form. They therefore admit SPA-CONV's scalar--polynomial accumulation.

\paragraph{Share-level composition.}
Jaguar's projection outputs can feed an attention protocol as additive shares, provided that fixed-point scaling is compatible and a common power-of-two share ring accommodates intermediate ranges. This avoids an additional ring conversion without requiring identical ciphertext moduli.
Public token permutations satisfy $\Pi(XW)=(\Pi X)W$; we verified this identity for the token permutations induced by MaxViT-T's block and grid partitioning~\cite{tu2022maxvit}, so these permutations can be applied locally to shares. The measurements below evaluate static-weight kernels; integration with protocols for the input-dependent attention core and end-to-end secure attention evaluation remain future work.

\paragraph{Dense attention projections.}
Both backends use $N=4096$, $p=2^{37}$, and nine fractional bits,
with identical synthetic integer inputs in $[-255,255]$ and nonzero
weights in $[-31,31]$.
Jaguar uses $q=2^{60}$; Cheetah uses a 60-bit NTT-prime data modulus
and a 109-bit key-level modulus.
Each retains its native packing.
On an Intel Xeon Gold 6250 with AVX enabled and one worker pinned
to the same core, we report medians of nine runs after two warm-ups,
alternating backend order.
Timing includes HE linear evaluation, reductions, and output-buffer
construction, excluding encoding, encryption/decryption, truncation,
masking, share export, and communication.

\paragraph{Results.}
Table~\ref{tab:attention_extension}(a) shows speedups of
$0.84$--$11.03\times$, with Jaguar faster on six of eight shapes.
The two 49-token shapes favor Cheetah's channel packing.
Weighting median latencies by graph occurrence counts gives
$\sum_s n_s t_s^{\mathrm{Cheetah}}/
 \sum_s n_s t_s^{\mathrm{Jaguar}}=2.24$.
All 1,505,280 logical output coefficients match the integer reference
exactly modulo $2^{37}$ after decryption.
These results measure projection kernels, not end-to-end MaxViT inference.

\section{Comparison with Recent Hybrid HE/2PC Systems}
\label{app:recent-hybrid-systems}

Recent systems such as Panther~\cite{feng2025panther} and SecONNds~\cite{balla2025seconnds} improve hybrid HE/2PC inference
through optimizations to linear computation and nonlinear protocols.
We compare their designs with Jaguar (Table~\ref{tab:recent-system-designs}) and examine how their nonlinear improvements complement Jaguar's exact HE-side truncation and reduced-bitwidth evaluation.

\begin{table}[h]
  \centering
  \caption{Comparison of representative hybrid HE/2PC inference systems.
  CrypTFlow2~\cite{rathee2020cryptflow2} refers to its HE-based linear backend, and SecONNds denotes
  the base variant rather than SecONNds-P.}
  \label{tab:recent-system-designs}
  \begingroup
  \footnotesize
  \setlength{\tabcolsep}{3pt}
  \renewcommand{\arraystretch}{1.15}
  \begin{tabularx}{\linewidth}{@{}l
    >{\raggedright\arraybackslash}X
    >{\raggedright\arraybackslash}X
    >{\raggedright\arraybackslash}X@{}}
    \toprule
    System & Linear & ReLU / comparison & Truncation \\
    \midrule
    CrypTFlow2
      & SIMD + rotations
      & Chosen-OT/LUT Millionaire
      & Faithful; bit-exact \\
    \addlinespace
    Cheetah
      & Coefficient encoding; no rotations
      & Same Millionaire + silent OT
      & One-bit approximate \\
    \addlinespace
    Panther
      & Shared ciphertext component; fewer polynomial multiplications
      & New OT Millionaire; $N$-way merging
      & Cheetah-style approximate \\
    \addlinespace
    SecONNds
      & NTT preprocessing; optional GPU
      & Boolean GMW/FMILL; no chosen-OT/LUT comparison
      & Cheetah-style approximate \\
    \addlinespace
    \textbf{Jaguar}
      & \textbf{Power-of-two SPA-CONV; no server NTT for convolution}
      & \textbf{Cheetah ReLU at lower bitwidth}
      & \textbf{Exact HE-side; no 2PC} \\
    \bottomrule
  \end{tabularx}
  \endgroup
\end{table}

\paragraph{Technical relationship and orthogonality.}
Panther, SecONNds, and Jaguar all target the polynomial multiplication and NTT costs in Cheetah~\cite{huang2022cheetah}, but at different levels.
SecONNds retains Cheetah's linear protocol and preprocesses weights into the NTT domain, keeping MACs in-domain until the final inverse transforms.
Panther also retains an NTT-based backend, but customizes HE so that multiple ciphertexts share one polynomial component, reducing polynomial multiplications, and reuses NTT representations across repeated inner products.
Jaguar instead changes the arithmetic substrate: SPA-CONV uses scalar--polynomial products, coefficient shifts, additions, and masks over a power-of-two ciphertext ring, eliminating server-side plaintext--ciphertext polynomial multiplication and NTT for convolution.
Their concrete linear kernels are therefore alternatives rather than direct add-ons, although their broader packing and reuse principles remain relevant.

Panther's shared-component construction reduces the number of polynomial multiplications by up to $2\times$ in the optimal case.
Jaguar achieves a measured $2.00$--$5.00\times$ speedup over Cheetah for the complete Conv+FC path across the evaluated models and LAN/WAN settings.

The nonlinear contributions are orthogonal. Jaguar does not propose a new ReLU or max-pooling protocol: it exports additive shares over $\mathbb{Z}_{2^\ell}$, the interface consumed by Panther and SecONNds.
Their nonlinear protocols can therefore replace Jaguar's Cheetah-based implementations.
Jaguar's contribution is to move exact truncation before nonlinear evaluation, reducing the input bitwidth rather than changing the nonlinear primitive itself.
This benefit is retained when the primitive is replaced; we directly verify this compatibility for SecONNds below.

CrypTFlow2~\cite{rathee2020cryptflow2} faithfully corrects truncation errors through interactive 2PC, whereas Cheetah permits a possible one-LSB error.
Panther and base SecONNds retain this approximate truncation while accelerating the underlying comparison and AND computations.
Jaguar instead performs exact local ciphertext-side shifts before HE-to-2PC conversion, using $2^f\mid\Delta$ and the correctness conditions established in our truncation analysis.
This removes post-ReLU 2PC truncation and allows the nonlinear protocol to operate on lower-bitwidth shares.

\paragraph{Matched evaluation and direct integration.}
We evaluate the public CPU implementation of SecONNds~\cite{seconnds_code} using the same hardware, four threads, fixed-point parameters, and emulated LAN/WAN settings as Jaguar.
Results use AVX disabled and measure online end-to-end inference; offline triple generation is excluded.
We evaluate both serial and low-round (LR) variants and report the faster one for each model and network setting.
Jaguar+SecONNds replaces only Jaguar's ReLU and max-pooling implementations with SecONNds, retaining SPA-CONV and exact HE-side truncation.
Speedup and communication reduction in Table~\ref{tab:seconnds-integration} are relative to Cheetah.

\begin{table}[t]
  \centering
  \caption{Matched online end-to-end evaluation with AVX disabled.
  Speedup and communication reduction are relative to Cheetah.
  Bold denotes the best measured result for each model and metric.
  A dash denotes an unmeasured configuration.}
  \label{tab:seconnds-integration}
  \begingroup
  \footnotesize
  \setlength{\tabcolsep}{3pt}
  \renewcommand{\arraystretch}{1.1}
  \resizebox{\linewidth}{!}{%
  \begin{tabular}{@{}llrrrrrr@{}}
    \toprule
    & & \multicolumn{2}{c}{LAN} & \multicolumn{2}{c}{WAN}
      & \multicolumn{2}{c}{Communication} \\
    \cmidrule(lr){3-4}\cmidrule(lr){5-6}\cmidrule(l){7-8}
    Model & Protocol & Latency (s) & Speedup
      & Latency (s) & Speedup & MiB & Reduction \\
    \midrule
    ResNet-18 & Cheetah
      & 30.30 & $1.00\times$ & 36.94 & $1.00\times$ & 366.94 & $1.00\times$ \\
    & SecONNds
      & 14.56 & $2.08\times$ & 39.08 & $0.95\times$ & 359.62 & $1.02\times$ \\
    & Jaguar
      & 8.64 & $3.51\times$ & \textbf{15.94} & $\mathbf{2.32\times}$ & 315.02 & $1.16\times$ \\
    & Jaguar+SecONNds
      & \textbf{5.74} & $\mathbf{5.28\times}$ & --- & --- & \textbf{308.73} & $\mathbf{1.19\times}$ \\
    \midrule
    ResNet-50 & Cheetah
      & 120.82 & $1.00\times$ & 142.72 & $1.00\times$ & 1335.63 & $1.00\times$ \\
    & SecONNds
      & 67.46 & $1.79\times$ & 143.87 & $0.99\times$ & 1377.76 & $0.97\times$ \\
    & Jaguar
      & 32.44 & $3.72\times$ & \textbf{51.50} & $\mathbf{2.77\times}$ & \textbf{962.00} & $\mathbf{1.38\times}$ \\
    & Jaguar+SecONNds
      & \textbf{24.31} & $\mathbf{4.97\times}$ & --- & --- & 998.84 & $1.34\times$ \\
    \midrule
    MobileNetV2 & Cheetah
      & 46.99 & $1.00\times$ & 67.25 & $1.00\times$ & 1042.62 & $1.00\times$ \\
    & SecONNds
      & 26.13 & $1.80\times$ & 84.10 & $0.80\times$ & 1100.53 & $0.95\times$ \\
    & Jaguar
      & 18.41 & $2.55\times$ & \textbf{32.42} & $\mathbf{2.07\times}$ & \textbf{592.51} & $\mathbf{1.76\times}$ \\
    & Jaguar+SecONNds
      & \textbf{10.80} & $\mathbf{4.35\times}$ & --- & --- & 625.42 & $1.67\times$ \\
    \bottomrule
  \end{tabular}%
  }
  \endgroup
\end{table}

Under LAN, Jaguar is $1.42$--$2.08\times$ faster than standalone SecONNds.
The integrated system further improves Jaguar by $1.51\times$, $1.33\times$, and $1.70\times$ on ResNet-18, ResNet-50, and MobileNetV2, respectively, confirming that SecONNds' nonlinear protocols compound with Jaguar's linear and truncation gains.
SecONNds is more sensitive to network latency because its nonlinear protocol requires more communication rounds than Cheetah.
Consequently, it is faster than Cheetah under our low-RTT LAN setting but slower under our WAN emulation.
Since standalone SecONNds is slower than Cheetah under our WAN setting, we evaluate Jaguar+SecONNds only under LAN.

Overall, Panther and SecONNds optimize the conventional linear backend and/or nonlinear protocol, whereas Jaguar provides a different power-of-two linear backend, exact zero-communication truncation, and narrower nonlinear inputs.
The direct SecONNds integration shows that recent nonlinear advances do not subsume Jaguar; they add further gains when the network environment favors them.

\section{Derivation of the Kernel-Level Complexity Table}
\label{app:conv_complexity_derivation}

Table~\ref{tab:conv_complexity} compares the dominant online server-side arithmetic of secure convolution kernels after translating different HE backends into estimated 64-bit multiplication counts.
The comparison is not intended to predict full end-to-end latency. It excludes key generation, encryption, decryption, HE-to-share export, communication, additions, coefficient shifts, masks, memory traffic, and other protocol-boundary costs unless explicitly shown. For the baseline systems, we use optimistic lower-envelope estimates under their native packing assumptions, omitting ceiling factors, tiling multipliers, preprocessing costs, and lower-order terms. This convention avoids overcharging baselines for costs that are not included in Jaguar's Scalar-Poly count.

\paragraph{Concrete 64-bit multiplication model.}
For prime-modulus RNS baselines, arithmetic is performed modulo RNS prime limbs. Thus, each NTT butterfly multiplication, pointwise product, and key-switching product requires a modular reduction. We model one modular product over a 60-bit RNS prime as
\[
    c_{\mathrm{red}} = 3
\]
64-bit multiplications, corresponding to one raw $64 \times 64$ product and two reduction-related products in a Barrett/Montgomery-style implementation. Additions, shifts, comparisons, and conditional subtractions are not counted.

Let
\[
    n_N = \log_2 N .
\]
For an NTT/RNS-based plaintext--ciphertext polynomial multiplication on an RLWE ciphertext with two polynomial components, we assume the plaintext polynomial can be pre-transformed whenever applicable. Hence, per ciphertext component and per RNS limb, we count one forward NTT on the ciphertext component, one pointwise product, and one inverse NTT:
\[
    \frac{N}{2}n_N + N + \frac{N}{2}n_N + N
    =
    N(n_N+2)
\]
modular products. Since an RLWE ciphertext has two polynomial components, one plaintext--ciphertext polynomial multiplication costs
\[
    \mathsf{PMult}_{64}
    =
    2c_{\mathrm{red}}NL(n_N+2)
    =
    N\Gamma_{\mathrm{mul}},
\]
where
\[
    \Gamma_{\mathrm{mul}}
    =
    6L(\log_2 N+2).
\]

For a homomorphic rotation/automorphism, we count the dominant key-switching work. The automorphism itself is an index permutation and is not counted as multiplication work. We model key switching as
$L^2$ NTTs of decomposed digits, $2L^2$ pointwise products with the switching key, and inverse NTTs for two output components over $L$ RNS limbs. This gives
\[
    \mathsf{Rot}_{64}
    =
    c_{\mathrm{red}}N
    \left[
        \left(\frac{n_N}{2}+2\right)L^2
        +
        (n_N+2)L
    \right]
    =
    N\Gamma_{\mathrm{rot}},
\]
where
\[
    \Gamma_{\mathrm{rot}}
    =
    3\left[
        \left(\frac{\log_2 N}{2}+2\right)L^2
        +
        (\log_2 N+2)L
    \right].
\]
This model is still favorable to the baselines: it does not separately charge basis extension, decomposition bookkeeping, lazy-reduction bookkeeping, additions, memory movement, or cache effects. In contrast, Jaguar uses a power-of-two ciphertext modulus, so coefficient reductions are realized by native wraparound, masking, or shifts rather than prime-modulus reduction. These non-multiplicative operations are not included in the Scalar-Poly count.

\paragraph{Gazelle dense convolution.}
Gazelle uses packed SIMD HE and ciphertext slot permutations for
homomorphic convolution~\cite{juvekar2018gazelle}. In
Table~\ref{tab:conv_complexity}, the HE-PMult column for Gazelle should be read as the cost of SIMD plaintext--ciphertext scalar multiplication translated into the same 64-bit multiplication model.

Let $c_n$ denote the number of channels packed into one ciphertext. We use the optimistic packing assumption
\[
    c_n \approx \frac{N}{HW}
\]
and suppress the small kernel factor $R^2$ as a constant for the dense convolution comparison. Under this assumption, the number of SIMD scalar multiplications is lower-bounded by
\[
    \frac{C_{\mathrm{in}}C_{\mathrm{out}}}{c_n}
    =
    \frac{HW C_{\mathrm{in}}C_{\mathrm{out}}}{N}.
\]
Multiplying by $\mathsf{PMult}_{64}=N\Gamma_{\mathrm{mul}}$ gives
\[
    HW C_{\mathrm{in}}C_{\mathrm{out}}\Gamma_{\mathrm{mul}}.
\]

For rotations, Gazelle has input-rotation and output-rotation variants. We use a lower-envelope estimate by taking the better native packing choice and suppressing constant kernel factors. After substituting
$c_n \approx N/(HW)$, the number of rotations is estimated as
\[
    \frac{HW(C_{\mathrm{in}}+C_{\mathrm{out}})}{N}
    +
    C_{\mathrm{out}}.
\]
Multiplying by $\mathsf{Rot}_{64}=N\Gamma_{\mathrm{rot}}$ gives
\[
    \big(HW(C_{\mathrm{in}}+C_{\mathrm{out}})
    + C_{\mathrm{out}}N\big)\Gamma_{\mathrm{rot}}.
\]

\paragraph{Cheetah dense convolution.}
Cheetah uses coefficient encoding to realize convolution through plaintext--ciphertext polynomial multiplication and avoids homomorphic rotations~\cite{huang2022cheetah}. Under the packed full-convolution view, the number of plaintext--ciphertext polynomial multiplications is
\[
    \frac{HW C_{\mathrm{in}}C_{\mathrm{out}}}{N}.
\]
Therefore, the estimated 64-bit multiplication count is
\[
    \frac{HW C_{\mathrm{in}}C_{\mathrm{out}}}{N}
    \cdot
    N\Gamma_{\mathrm{mul}}
    =
    HW C_{\mathrm{in}}C_{\mathrm{out}}\Gamma_{\mathrm{mul}}.
\]
Since Cheetah's coefficient encoding removes homomorphic rotations from the convolution kernel, the HE-Rot entry is zero.

\paragraph{Rhombus pointwise convolution.}
A pointwise convolution with output spatial size $S_{\mathrm{out}}=H'W'$ can be viewed
as a matrix multiplication between
\[
    X' \in \mathbb{Z}^{S_{\mathrm{out}} \times C_{\mathrm{in}}}
    \quad \text{and} \quad
    W' \in \mathbb{Z}^{C_{\mathrm{in}} \times C_{\mathrm{out}}}.
\]
Thus, the total matrix-multiplication volume is
$S_{\mathrm{out}} C_{\mathrm{in}}C_{\mathrm{out}}$. Rhombus-V2 reduces the number of plaintext--ciphertext multiplications to
\[
    \frac{S_{\mathrm{out}} C_{\mathrm{in}}C_{\mathrm{out}}}{N}
\]
and the number of rotations to
\[
    \sqrt{\frac{S_{\mathrm{out}} C_{\mathrm{in}}C_{\mathrm{out}}}{N}}
\]
using input-output packing and split-point picking~\cite{he2024Rhombus}.
Therefore, the translated PMult cost is
\[
    \frac{S_{\mathrm{out}} C_{\mathrm{in}}C_{\mathrm{out}}}{N}
    \cdot
    N\Gamma_{\mathrm{mul}}
    =
    S_{\mathrm{out}} C_{\mathrm{in}}C_{\mathrm{out}}\Gamma_{\mathrm{mul}},
\]
and the translated rotation cost is
\[
    \sqrt{\frac{S_{\mathrm{out}} C_{\mathrm{in}}C_{\mathrm{out}}}{N}}
    \cdot
    N\Gamma_{\mathrm{rot}}
    =
    \sqrt{S_{\mathrm{out}} C_{\mathrm{in}}C_{\mathrm{out}}N}\Gamma_{\mathrm{rot}}.
\]

\paragraph{Falcon depthwise convolution.}
Falcon optimizes depthwise convolution by packing multiple depthwise filters/channels into one weight polynomial and by applying communication-aware operator tiling~\cite{xu2023falcon}. Let $C_w$ denote the number of channels packed into one weight polynomial.
Ignoring ceiling factors and assuming a single spatial tile, the number of plaintext--ciphertext polynomial multiplications is
\[
    \frac{C_{\mathrm{out}}}{C_w}.
\]
Thus, the estimated 64-bit multiplication count is
\[
    \frac{C_{\mathrm{out}}}{C_w}
    \cdot
    N\Gamma_{\mathrm{mul}}
    =
    \frac{C_{\mathrm{out}}}{C_w}N\Gamma_{\mathrm{mul}}.
\]
Falcon does not require homomorphic rotations in this depthwise coefficient-packing path, so the HE-Rot entry is zero. If spatial tiling is used, the above PMult term is multiplied by the number of spatial tiles. We omit this factor in Table~\ref{tab:conv_complexity} to keep the baseline estimate optimistic.

\section{Detailed Protocols and Algorithm}

\subsection{SPA-Conv: Scalar--Polynomial Accumulation Convolution}
\label{app:SPA-Conv_algo}

We describe Jaguar's linear convolution kernel as \emph{Scalar--Polynomial Accumulation Convolution} (SPA-Conv).
The key idea is to evaluate convolution directly in the coefficient domain by combining scalar--polynomial products with shifted accumulations in the power-of-two ciphertext ring.

Let the input tensor have shape $(H,W,C_{\mathrm{in}})$ and the filter tensor have shape $(K,K,C_{\mathrm{in}},C_{\mathrm{out}})$, where $K=2r+1$ is odd.
Let
\[
W_{\mathrm{pp}} = W + p_l + p_r
\]
denote the padded spatial width and
\[
S = W_{\mathrm{pp}}^2
\]
the padded spatial footprint of one input channel.

SPA-Conv operates on a collection of encrypted input blocks
\[
\mathcal{X} = \{\mathsf{CT}_t\}_{t=0}^{M-1},
\]
where each $\mathsf{CT}_t$ is an RLWE ciphertext over the power-of-two ring $R_q=\mathbb{Z}_q[X]/(X^N+1)$.
Depending on the packing regime, a ciphertext may contain multiple padded input channels or a spatial strip of a single padded channel.

We define a layout map
\[
\Lambda(c,j) = (\iota,\sigma),
\]
which returns the source ciphertext index $\iota$ and the in-ciphertext coefficient offset $\sigma$ for input channel $c$ and output block index $j$.
We use $B$ to denote the number of ciphertext blocks produced per output channel.
In the no-split regime, $B=1$.
In the split regime, $B=\nu$, where $\nu$ is the number of spatial strips used to represent one padded channel.

For dense and pointwise convolutions, the contributing input-channel set is
\[
\mathcal{C}(o) = [C_{\mathrm{in}}],
\]
for each output channel $o$.
For depthwise convolutions, the contributing input-channel set is
\[
\mathcal{C}(o) = \{o\}.
\]
We further define the kernel-offset set
\[
\Omega_K = \{-r,\ldots,r\}\times\{-r,\ldots,r\},
\]
and the coefficient-domain offset function
\[
\delta(u,v)=uW_{\mathrm{pp}}+v.
\]

\paragraph{Packing layouts.}
A convenient abstraction is:
\[
\Lambda(c,0)=
\Bigl(
\bigl\lfloor c/\eta \bigr\rfloor,\;
(c\bmod \eta)\,S
\Bigr),
\qquad
\eta=\left\lfloor \frac{N}{S} \right\rfloor,
\]
in the packed regime $S\le N$, and
\[
\Lambda(c,j)=(cB+j,\,0),
\qquad j\in[B],
\]
in the split regime $S>N$.

For compactness, define
\[
\mathcal{I}_{o,j} = \mathcal{C}(o) \times \Omega_K,
\]
namely the set of all input-channel and spatial-offset pairs contributing to output channel $o$ and block $j$.

\begin{algorithm}[H]
\caption{SPA-Conv over a power-of-two ciphertext ring}
\label{alg:SPA-Conv}
\begin{algorithmic}[1]
\STATE \textbf{Input:} encrypted input blocks $\mathcal{X}=\{\mathsf{CT}_t\}$, filter tensor $\mathbf{W}$, layout map $\Lambda$, padded width $W_{\mathrm{pp}}$, kernel size $K=2r+1$, truncation amount $\tau$
\STATE \textbf{Output:} truncated encrypted output blocks $\{\mathsf{Y}_{o,j}\}$

\FOR{$o=0$ to $C_{\mathrm{out}}-1$}
    \FOR{$j=0$ to $B-1$}
        \STATE choose an anchor term $(c^\star,u^\star,v^\star)$
        \STATE $(\iota^\star,\sigma^\star) \gets \Lambda(c^\star,j)$
        \STATE $\eta^\star \gets \sigma^\star + \delta(u^\star,v^\star)$
        \STATE $\mathsf{Y}_{o,j} \gets \mathsf{InitTerm}\!\bigl(W[u^\star,v^\star,c^\star,o], \mathsf{CT}_{\iota^\star}, s^\star\bigr)$
        \FORALL{$(c,u,v) \in \mathcal{I}_{o,j} \setminus \{(c^\star,u^\star,v^\star)\}$}
            \STATE $(\iota,\sigma) \gets \Lambda(c,j)$
            \STATE $s \gets \sigma + \delta(u,v)$
            \STATE $\mathsf{Y}_{o,j} \gets \mathsf{AccTerm}\!\bigl(\mathsf{Y}_{o,j}, W[u,v,c,o], \mathsf{CT}_{\iota}, s\bigr)$
        \ENDFOR
        \STATE $\mathsf{Y}_{o,j} \gets \mathsf{Trunc}_{\tau}(\mathsf{Y}_{o,j})$
        \STATE $\mathsf{Y}_{o,j} \gets \mathsf{MaskReduce}(\mathsf{Y}_{o,j})$
    \ENDFOR
\ENDFOR
\STATE \textbf{return} $\{\mathsf{Y}_{o,j}\}$
\end{algorithmic}
\end{algorithm}

\paragraph{Primitive operations.}
$\mathsf{InitTerm}(\alpha,\mathsf{CT},s)$ initializes an accumulator with the shifted scalar term $\alpha \cdot \mathrm{Shift}_s(\mathsf{CT})$.
$\mathsf{AccTerm}(\mathsf{Y},\alpha,\mathsf{CT},s)$ adds the shifted scalar term $\alpha \cdot \mathrm{Shift}_s(\mathsf{CT})$ into the accumulator $\mathsf{Y}$.
Here $\mathrm{Shift}_s(\cdot)$ denotes the negacyclic coefficient shift in $R_q$.
$\mathsf{Trunc}_{\tau}$ is Jaguar's exact ciphertext-side truncation from Section~3.1, and $\mathsf{MaskReduce}$ denotes optional local masking / reduction before the subsequent boundary conversion.

\paragraph{Special cases.}
The same algorithm specializes to several common CNN operators:
\begin{itemize}
    \item \textbf{Pointwise $1\times1$ convolution:} set $K=1$, so $\Omega_K=\{(0,0)\}$.
    \item \textbf{Dense $3\times3$ convolution:} set $K=3$ and 
    $\mathcal{C}(o)=[C_{\mathrm{in}}]$.
    \item \textbf{Depthwise $3\times3$ convolution:} set $K=3$ and $\mathcal{C}(o)=\{o\}$.
\end{itemize}
Thus, SPA-Conv provides a single arithmetic skeleton for pointwise, dense, and depthwise convolutions.

\subsection{Client-Side Grouped Reconstruction}
\label{app:grouped_reconstruction}
\begin{algorithm}[h]
\caption{Client-side grouped reconstruction}
\label{alg:grouped_reconstruct}
\begin{algorithmic}[1]
\REQUIRE Grouped exports $\mathcal{G}=\{G_\gamma\}$, secret key $\hat{s}(X)$, parameters $(Q,P,f)$
\ENSURE Client additive share $\langle\mathbf{T}'\rangle_C$

\STATE Initialize $\langle\mathbf{T}'\rangle_C$ with zeros
\STATE Precompute an NTT representation of $\hat{s}(X)$ over an auxiliary NTT prime $q_{\mathrm{ntt}}$
\FORALL{$G_\gamma=(\hat{a}_\gamma,\mathbf{b}_\gamma,\boldsymbol{\kappa}_\gamma,
\mathbf{ch}_\gamma,\mathbf{off}_\gamma,\beta_\gamma)\in\mathcal{G}$}
    \STATE Compute
    \[
        \hat{d}_\gamma(X) \leftarrow \hat{a}_\gamma(X)\hat{s}(X) \pmod{X^N+1}
    \]
    using the auxiliary NTT prime and map the result back to residues modulo $2^Q$
    \FOR{$i=0$ to $|\mathbf{b}_\gamma|-1$}
        \STATE $\kappa \leftarrow \boldsymbol{\kappa}_\gamma[i]$
        \STATE $u \leftarrow \hat{d}_\gamma[\kappa]+\mathbf{b}_\gamma[i]\pmod{2^Q}$
        \STATE
        \[
            z \leftarrow
            \left(
            (u+2^{Q-P-1-f}) \bmod 2^Q
            \right)\gg (Q-P)
            \pmod{2^{P-f}}
        \]
        \STATE Place $z$ into channel $\mathbf{ch}_\gamma[i]$ and output position $\mathbf{off}_\gamma[i]$
    \ENDFOR
\ENDFOR
\RETURN $\langle\mathbf{T}'\rangle_C$
\end{algorithmic}
\end{algorithm}

The grouped export protocol shares the RLWE linear component across many exported coefficients, but each coefficient still has a coefficient-specific decryption equation. This subsection makes the reconstruction rule explicit.

Let a truncated source RLWE ciphertext be
\[
    Y_\gamma=(\hat{b}_\gamma(X),\hat{a}_\gamma(X))\in R_{2^Q}^2,
\]
and let the secret key be \(\hat{s}(X)\). The decrypted polynomial before fixed-point decoding is
\[
    \hat{u}_\gamma(X)=\hat{b}_\gamma(X)+\hat{a}_\gamma(X)\hat{s}(X)\pmod{X^N+1,2^Q}.
\]
For a logical output value \((o,\xi)\), the extraction map gives
\[
    \Phi(o,\xi)=(\gamma,\kappa),
\]
so the client must reconstruct the coefficient
\[
    \hat{u}_{\gamma,\kappa}
    =
    \hat{b}_\gamma[\kappa]+(\hat{a}_\gamma \hat{s})[\kappa]\pmod{2^Q}.
\]
Therefore, the grouped export record carries the coefficient indices
\(\boldsymbol{\kappa}_\gamma\), and the client reconstructs the polynomial
\(\hat{a}_\gamma(X)\hat{s}(X)\) once per source ciphertext group. The resulting coefficient
\((\hat{a}_\gamma \hat{s})[\kappa_i]\) is then reused with the corresponding exported
constant term \(\mathbf{b}_\gamma[i]\).

In our implementation, this multiplication is performed locally by the client using an auxiliary NTT prime \(q_{\mathrm{ntt}}\). This auxiliary NTT is only an implementation technique for client-side reconstruction; it does not change the Jaguar ciphertext modulus, the protocol messages, or the communication pattern. The client first maps the coefficients of \(a_\gamma\) and the ternary secret key to residues modulo \(q_{\mathrm{ntt}}\), computes
\[
    \hat{d}_\gamma(X)=\hat{a}_\gamma(X)\hat{s}(X)\pmod{X^N+1}
\]
via NTT multiplication, and maps the selected coefficients back to residues modulo \(2^Q\). For each exported coefficient \(i\), it then computes
\[
    u_i =
    \hat{d}_\gamma[\boldsymbol{\kappa}_\gamma[i]]
    +
    \mathbf{b}_\gamma[i]
    \pmod{2^Q}.
\]
Finally, because the ciphertext has already undergone \(f\)-bit HE-side truncation, the client applies the same two-step fixed-point decoding rule as Algorithm~\ref{alg:dec}:
\[
    z_i =
    \left(
    (u_i+2^{Q-P-1-f}) \bmod 2^Q
    \right)\gg (Q-P)
    \pmod{2^{P-f}}.
\]
The decoded value \(z_i\) is written to the output location specified by
\((\mathbf{ch}_\gamma[i],\mathbf{off}_\gamma[i])\). Thus, the grouped export shares the RLWE linear component in communication, while reconstruction remains coefficient-correct through the explicit coefficient index
\(\boldsymbol{\kappa}_\gamma[i]\).

\paragraph{Auxiliary NTT modulus.}
The auxiliary prime must be large enough so that the integer coefficients of
\(\hat{a}_\gamma(X)\hat{s}(X)\) do not wrap modulo \(q_{\mathrm{ntt}}\). Let \(h=\|\hat{s}\|_0\)
be the number of nonzero secret-key coefficients. After \(f\)-bit ciphertext-side truncation, the exported linear component has effective magnitude below
\(2^{Q-f}\). A conservative sufficient condition is
\[
    q_{\mathrm{ntt}}/2 > h\cdot 2^{Q-f}.
\]
For our main setting \(N=2048\), \(Q=54\), and \(f=9\), we have
\(h\le N=2^{11}\), so the product bound is at most \(2^{56}\). A 60-bit NTT-friendly prime has centered range about \(2^{59}\), which is sufficient for the grouped reconstruction path. We choose a 60-bit prime $q_{\mathrm{ntt}}\equiv 1 \pmod{2N}$ and lift
coefficients to centered representatives before the auxiliary-prime NTT product. If the export is performed without HE-side truncation (\(f=0\)), this single-prime bound is no longer sufficient; in that case one can either use a multi-prime CRT NTT or a direct power-of-two negacyclic multiplication. Jaguar's evaluated protocol performs grouped export after HE-side truncation.

\section{Complexity Analysis of Linear-Layer Protocol of Jaguar}

\subsection{Complexity Analysis of SPA-Conv}
\label{app:SPA-Conv_complexity}

We analyze SPA-Conv at the level of coefficient-wise integer operations.
Our goal is to expose the arithmetic structure of Jaguar's linear convolution kernel independently of any specific code path.

\paragraph{Cost model.}
An RLWE ciphertext in Jaguar consists of two degree-$N$ polynomials.
We count five types of low-level integer operations:
\emph{multiplications}, \emph{additions/subtractions}, \emph{arithmetic shifts}, \emph{bitwise masks}, and \emph{plaintext-mask additions}.
We ignore memory writes and index arithmetic, since they do not affect the leading arithmetic cost.

For one ciphertext block, the SPA-Conv primitives have the following costs:
\begin{itemize}
    \item $\mathsf{InitTerm}$ touches both RLWE components over all $N$ coefficients:
    \[
    \mathrm{Mul}(\mathsf{InitTerm}) = 2N.
    \]
    \item $\mathsf{AccTerm}$ performs one shifted scalar accumulation on both RLWE components.
    It requires one coefficient-wise scalar multiplication and one accumulator update per coefficient, together with sign-adjustments induced by negacyclic wrap-around.
    We conservatively upper-bound its arithmetic cost by
    \[
    \mathrm{Mul}(\mathsf{AccTerm}) = 2N,
    \qquad
    \mathrm{Add/Sub}(\mathsf{AccTerm}) \le 4N.
    \]
    \item $\mathsf{Trunc}_{\tau}$ performs one arithmetic right shift on each coefficient of both ciphertext components:
    \[
    \mathrm{Shift}(\mathsf{Trunc}_{\tau}) = 2N.
    \]
    \item $\mathsf{MaskReduce}$ applies power-of-two reduction to both ciphertext components:
    \[
    \mathrm{Mask}(\mathsf{MaskReduce}) = 2N.
    \]
    \item If a plaintext mask is added before export, only the first ciphertext component is updated:
    \[
    \mathrm{PlainAdd} = N.
    \]
\end{itemize}

\paragraph{Number of convolution terms.}
For one output channel $o$ and one output block $b$, define
\[
T_o = |\mathcal{C}(o)|\,K^2.
\]
This is the total number of shifted scalar terms accumulated into that output block.
SPA-Conv uses one anchor term to initialize the accumulator and then accumulates the remaining $T_o-1$ terms.

Therefore, the arithmetic cost for one output block is
\[
\mathrm{Mul}_{o,j}
=
2N + 2N(T_o-1)
=
2N\,T_o,
\]
\[
\mathrm{Add/Sub}_{o,j}
\le
4N(T_o-1),
\]
and the local post-processing cost is
\[
\mathrm{Shift}_{o,j}=2N,
\qquad
\mathrm{Mask}_{o,j}=2N,
\qquad
\mathrm{PlainAdd}_{o,j}=N.
\]

\paragraph{General layer complexity.}
Let $B$ denote the number of ciphertext blocks produced per output channel.
Then the total arithmetic cost of one SPA-Conv layer is
\[
\mathrm{Mul}_{\mathrm{SPA-Conv}}
=
2N B \sum_{o=0}^{C_{\mathrm{out}}-1} T_o
=
2N B K^2 \sum_{o=0}^{C_{\mathrm{out}}-1} |\mathcal{C}(o)|,
\]
\[
\mathrm{Add/Sub}_{\mathrm{SPA-Conv}}
\le
4N B \sum_{o=0}^{C_{\mathrm{out}}-1} (T_o-1)
=
4N B \left(
K^2 \sum_{o=0}^{C_{\mathrm{out}}-1} |\mathcal{C}(o)|
-
C_{\mathrm{out}}
\right),
\]
\[
\mathrm{Shift}_{\mathrm{SPA-Conv}} = 2N B C_{\mathrm{out}},
\qquad
\mathrm{Mask}_{\mathrm{SPA-Conv}} = 2N B C_{\mathrm{out}},
\qquad
\mathrm{PlainAdd}_{\mathrm{SPA-Conv}} = N B C_{\mathrm{out}}.
\]

These expressions separate the convolution kernel from the boundary conversion.
In particular, the grouped output-export procedure used by Jaguar is linear in the number of exported coefficients and depends on the chosen export format; we therefore treat it separately from the arithmetic kernel analyzed here.

\paragraph{Dense and depthwise specializations.}
For dense and pointwise convolutions,
\[
|\mathcal{C}(o)| = C_{\mathrm{in}}
\qquad \forall o,
\]
and therefore
\[
\mathrm{Mul}_{\mathrm{dense}}
=
2N B C_{\mathrm{out}} C_{\mathrm{in}} K^2,
\]
\[
\mathrm{Add/Sub}_{\mathrm{dense}}
\le
4N B C_{\mathrm{out}} (C_{\mathrm{in}}K^2 - 1),
\]
\[
\mathrm{Shift}_{\mathrm{dense}} = 2N B C_{\mathrm{out}},
\quad
\mathrm{Mask}_{\mathrm{dense}} = 2N B C_{\mathrm{out}},
\quad
\mathrm{PlainAdd}_{\mathrm{dense}} = N B C_{\mathrm{out}}.
\]

For depthwise convolutions,
\[
|\mathcal{C}(o)| = 1
\qquad \forall o,
\]
and therefore
\[
\mathrm{Mul}_{\mathrm{dw}}
=
2N B C_{\mathrm{out}} K^2,
\]
\[
\mathrm{Add/Sub}_{\mathrm{dw}}
\le
4N B C_{\mathrm{out}} (K^2 - 1),
\]
\[
\mathrm{Shift}_{\mathrm{dw}} = 2N B C_{\mathrm{out}},
\quad
\mathrm{Mask}_{\mathrm{dw}} = 2N B C_{\mathrm{out}},
\quad
\mathrm{PlainAdd}_{\mathrm{dw}} = N B C_{\mathrm{out}}.
\]

\paragraph{Common special cases.}
The formulas above directly yield the following cases of practical interest.

\subparagraph{Pointwise $1\times1$ convolution.}
Setting $K=1$ in the dense formula gives
\[
\mathrm{Mul}_{1\times1}
=
2N B C_{\mathrm{out}} C_{\mathrm{in}},
\]
\[
\mathrm{Add/Sub}_{1\times1}
\le
4N B C_{\mathrm{out}} (C_{\mathrm{in}} - 1),
\]
\[
\mathrm{Shift}_{1\times1}=2N B C_{\mathrm{out}},
\quad
\mathrm{Mask}_{1\times1}=2N B C_{\mathrm{out}},
\quad
\mathrm{PlainAdd}_{1\times1}=N B C_{\mathrm{out}}.
\]

\subparagraph{Dense $3\times3$ convolution.}
Setting $K=3$ in the dense formula gives
\[
\mathrm{Mul}_{3\times3\text{-dense}}
=
18N B C_{\mathrm{out}} C_{\mathrm{in}},
\]
\[
\mathrm{Add/Sub}_{3\times3\text{-dense}}
\le
4N B C_{\mathrm{out}} (9C_{\mathrm{in}} - 1),
\]
\[
\mathrm{Shift}_{3\times3\text{-dense}}=2N B C_{\mathrm{out}},
\quad
\mathrm{Mask}_{3\times3\text{-dense}}=2N B C_{\mathrm{out}},
\quad
\mathrm{PlainAdd}_{3\times3\text{-dense}}=N B C_{\mathrm{out}}.
\]

\subparagraph{Depthwise $3\times3$ convolution.}
Setting $K=3$ in the depthwise formula gives
\[
\mathrm{Mul}_{3\times3\text{-dw}}
=
18N B C_{\mathrm{out}},
\]
\[
\mathrm{Add/Sub}_{3\times3\text{-dw}}
\le
32N B C_{\mathrm{out}},
\]
\[
\mathrm{Shift}_{3\times3\text{-dw}}=2N B C_{\mathrm{out}},
\quad
\mathrm{Mask}_{3\times3\text{-dw}}=2N B C_{\mathrm{out}},
\quad
\mathrm{PlainAdd}_{3\times3\text{-dw}}=N B C_{\mathrm{out}}.
\]

\paragraph{Interpretation.}
These expressions make the arithmetic structure of Jaguar explicit.
Dense $K\times K$ convolution scales as $\Theta(NBC_{\mathrm{out}}C_{\mathrm{in}}K^2)$, pointwise convolution scales as $\Theta(NBC_{\mathrm{out}}C_{\mathrm{in}})$, and depthwise convolution scales as $\Theta(NBC_{\mathrm{out}}K^2)$.
Hence, relative to dense $3\times3$ convolution, pointwise convolution removes the spatial $K^2$ factor, whereas depthwise convolution removes the multiplicative $C_{\mathrm{in}}$ factor.
The packing layout only changes the block multiplier $B$: in the no-split regime, $B=1$, while in the split regime, $B=\nu$.

\subsection{Complexity Analysis of Grouped Output Export}

\paragraph{Grouped Output Export.}
\label{app:grouped_export}

After SPA-Conv produces encrypted output blocks $\{\mathsf{Y}_{o,j}\}$, Jaguar converts them into additive shares by exporting only the coefficients that correspond to valid convolution outputs.
For each source RLWE ciphertext, Jaguar extracts the valid output coefficients as coefficient-wise LWE ciphertexts of the form $(b_\kappa, \hat{a})$, where all extracted coefficients from the same source ciphertext share the same linear component $\hat{a}$.
Instead of transmitting these LWEs independently, Jaguar packs them into a grouped export record consisting of one shared $\hat{a}$, multiple $b$ values, and output-location metadata.

Let
\[
\Phi : (o,\xi) \mapsto (\gamma,\kappa)
\]
denote the extraction map induced by the packing layout, where
$o$ is the output channel index, $\xi$ is the logical output position,
$\gamma$ is the source RLWE ciphertext index, and $\kappa$ is the coefficient index inside that ciphertext.

For each source ciphertext $\mathsf{Y}_\gamma = (b_\gamma, \hat{a}_\gamma)$, define the extracted set
\[
\mathcal{P}_\gamma
=
\{(o,\xi,\kappa) : \Phi(o,\xi) = (\gamma,\kappa)\}.
\]
Jaguar forms one grouped export per non-empty set $\mathcal{P}_\gamma$.


\begin{algorithm}[H]
\caption{Grouped output export}
\label{alg:grouped_export}
\begin{algorithmic}[1]
\REQUIRE Truncated encrypted output blocks $\{Y_\gamma\}$ and extraction map $\Phi$
\ENSURE Grouped exports $\mathcal{G}=\{G_\gamma\}$

\STATE Initialize $\mathcal{G}\leftarrow\emptyset$
\FORALL{source ciphertext indices $\gamma$ such that $P_\gamma\neq\emptyset$}
    \STATE Let $Y_\gamma=(b_\gamma,\hat{a}_\gamma)$
    \STATE Initialize empty vectors $\mathbf{b}_\gamma,\boldsymbol{\kappa}_\gamma,\mathbf{ch}_\gamma,\mathbf{off}_\gamma$
    \FORALL{$(o,\xi,\kappa)\in P_\gamma$}
        \STATE Append $b_\gamma[\kappa]$ to $\mathbf{b}_\gamma$
        \STATE Append $\kappa$ to $\boldsymbol{\kappa}_\gamma$
        \STATE Append $o$ to $\mathbf{ch}_\gamma$
        \STATE Append $\xi$ to $\mathbf{off}_\gamma$
    \ENDFOR
    \STATE Choose transport bit-width $\beta_\gamma$
    \STATE $G_\gamma\leftarrow
    (\hat{a}_\gamma,\mathbf{b}_\gamma,\boldsymbol{\kappa}_\gamma,
    \mathbf{ch}_\gamma,\mathbf{off}_\gamma,\beta_\gamma)$
    \STATE Append $G_\gamma$ to $\mathcal{G}$
\ENDFOR
\RETURN $\mathcal{G}$
\end{algorithmic}
\end{algorithm}

\paragraph{Complexity of Grouped Output Export.}
\label{app:grouped_export_complexity}

We now analyze the boundary-conversion cost of grouped output export.
This cost is separate from the SPA-Conv arithmetic kernel and depends on how many logical output coefficients are exported from each source ciphertext.

\paragraph{Notation.}
Let $G$ denote the number of non-empty export groups.
For group $g \in [G]$, let
\[
m_g
\]
be the number of exported coefficients in that group, and let
\[
\beta_g
\]
be the transport bit-width used to pack its values.
Let
\[
E = \sum_{g=1}^{G} m_g
\]
denote the total number of exported output coefficients.
For convolution, $E$ equals the total number of logical output values, i.e.,
\[
E = C_{\mathrm{out}} H' W'.
\]

\paragraph{Server-side grouping cost.}
For each group, the server copies one shared RLWE linear component of length $N$ and extracts $m_g$ constant coefficients.
Hence, the arithmetic work of forming all groups is
\[
\Theta\!\left(GN + E\right)
\]
coefficient reads/writes.
If bit-packing is included, the bit-level packing work is
\[
\Theta\!\left(\sum_{g=1}^{G} \beta_g (N + m_g)\right).
\]

\paragraph{Communication cost.}
Group \(g\) contains: (i) one shared RLWE linear component of length \(N\),
(ii) \(m_g\) exported constant terms, (iii) \(m_g\) coefficient indices, and
(iv) metadata describing the logical output locations of those \(m_g\) values.
Thus, the transmitted payload of group \(g\) is
\[
    \beta_g(N+m_g)
    +
    m_g\left(
    \lceil\log_2 N\rceil
    +
    \lceil\log_2 C_{\mathrm{out}}\rceil
    +
    \lceil\log_2(H'W')\rceil
    \right)
\]
bits, up to lower-order header terms. Summing over all groups,
\[
    \mathrm{Comm}_{\mathrm{grouped}}
    =
    \sum_{g=1}^{G}
    \left[
    \beta_g(N+m_g)
    +
    m_g\left(
    \lceil\log_2 N\rceil
    +
    \lceil\log_2 C_{\mathrm{out}}\rceil
    +
    \lceil\log_2(H'W')\rceil
    \right)
    \right]
    +O(G).
\]

\paragraph{Client-side reconstruction cost.}
For each group, the client reconstructs the negacyclic product
\(a_\gamma(X)s(X)\) once and then selects the \(m_g\) requested coefficients.
With the auxiliary NTT implementation, the cost is
\[
    \Theta(GN\log N + E)
\]
arithmetic operations, after one global precomputation of the NTT representation of the secret key. The \(M\) term accounts for coefficient selection, addition of the exported constant terms, and power-of-two decode/round steps.
This client-side cost is still separate from the server-side \textsc{SPA-Conv}
kernel and does not affect communication.

A direct power-of-two implementation could instead compute the same product by negacyclic add/sub accumulation using the ternary secret key; the protocol is unchanged, but the implementation cost becomes \(O(GNh+E)\), where
\(h=\|\hat{s}\|_0\).

\paragraph{Comparison with per-output export.}
Without grouping, exporting each logical output independently would require one length-$N$ linear component per output value, leading to communication proportional to $EN$.
Grouped export reduces this to $GN + E$, where typically $G \ll E$ because many output coefficients originate from the same source RLWE ciphertext.
This is the main reason Jaguar's HE-to-share boundary remains efficient even when the logical convolution output is dense.





\newpage
\section*{NeurIPS Paper Checklist}

\begin{enumerate}

\item {\bf Claims}
    \item[] Question: Do the main claims made in the abstract and introduction accurately reflect the paper's contributions and scope?
    \item[] Answer: \answerYes{} 
    \item[] Justification: The main claims in the abstract and introduction match the proposed techniques and experimental results.

\item {\bf Limitations}
    \item[] Question: Does the paper discuss the limitations of the work performed by the authors?
    \item[] Answer: \answerYes{} 
    \item[] Justification: Section~\ref{sec:conv_protocol} discusses Jaguar's current FC/MatMul limitation, and Appendix~\ref{app:limitation} provides further discussion of native power-of-two FC/MatMul extensions.

\item {\bf Theory assumptions and proofs}
    \item[] Question: For each theoretical result, does the paper provide the full set of assumptions and a complete (and correct) proof?
    \item[] Answer: \answerYes{} 
    \item[] Justification: The main correctness theorem states its assumptions and is proved in Appendix~\ref{app:proof}. Security/parameter assumptions and the protocol/cost analyses are also provided in Appendices.

    \item {\bf Experimental result reproducibility}
    \item[] Question: Does the paper fully disclose all the information needed to reproduce the main experimental results of the paper to the extent that it affects the main claims and/or conclusions of the paper (regardless of whether the code and data are provided or not)?
    \item[] Answer: \answerYes{} 
    \item[] Justification: The paper reports the evaluated models, parameters, baselines, hardware, network settings. Source code and reproduction scripts are publicly available at \url{https://github.com/yellowon/Jaguar}.

\item {\bf Open access to data and code}
    \item[] Question: Does the paper provide open access to the data and code, with sufficient instructions to faithfully reproduce the main experimental results, as described in supplemental material?
    \item[] Answer: \answerYes{} 
    \item[] Justification: Source code and reproduction scripts are publicly available at \url{https://github.com/yellowon/Jaguar}.

\item {\bf Experimental setting/details}
    \item[] Question: Does the paper specify all the training and test details (e.g., data splits, hyperparameters, how they were chosen, type of optimizer) necessary to understand the results?
    \item[] Answer: \answerYes{} 
    \item[] Justification: Section 5.1 specifies the HE parameters, models, dataset, implementation, hardware, network settings, and baselines.

\item {\bf Experiment statistical significance}
    \item[] Question: Does the paper report error bars suitably and correctly defined or other appropriate information about the statistical significance of the experiments?
    \item[] Answer: \answerNo{} 
    \item[] Justification: We do not report confidence intervals or error bars because the main evaluation requires expensive end-to-end private inference runs; the code is provided for reproducibility.

\item {\bf Experiments compute resources}
    \item[] Question: For each experiment, does the paper provide sufficient information on the computer resources (type of compute workers, memory, time of execution) needed to reproduce the experiments?
    \item[] Answer: \answerYes{} 
    \item[] Justification: The hardware configuration and execution settings are described in Section 5.1.
    
\item {\bf Code of ethics}
    \item[] Question: Does the research conducted in the paper conform, in every respect, with the NeurIPS Code of Ethics \url{https://neurips.cc/public/EthicsGuidelines}?
    \item[] Answer: \answerYes{} 
    \item[] Justification: The work follows the NeurIPS Code of Ethics and uses existing public datasets, pretrained models, and public baseline implementations.

\item {\bf Broader impacts}
    \item[] Question: Does the paper discuss both potential positive societal impacts and negative societal impacts of the work performed?
    \item[] Answer: \answerNo{} 
    \item[] Justification: The paper focuses on the technical efficiency of private CNN inference and does not include a separate discussion of broader societal impacts. The main expected benefit is improved privacy-preserving inference, while deployment-specific risks should be assessed in downstream applications.
    
\item {\bf Safeguards}
    \item[] Question: Does the paper describe safeguards that have been put in place for responsible release of data or models that have a high risk for misuse (e.g., pre-trained language models, image generators, or scraped datasets)?
    \item[] Answer: \answerNA{} 
    \item[] Justification: The paper does not release high-risk models, scraped datasets, or generative systems requiring additional safeguards.

\item {\bf Licenses for existing assets}
    \item[] Question: Are the creators or original owners of assets (e.g., code, data, models), used in the paper, properly credited and are the license and terms of use explicitly mentioned and properly respected?
    \item[] Answer: \answerYes{} 
    \item[] Justification: Existing datasets, pretrained models, libraries, and baseline implementations are cited and documented in the paper or supplementary material.

\item {\bf New assets}
    \item[] Question: Are new assets introduced in the paper well documented and is the documentation provided alongside the assets?
    \item[] Answer: \answerYes{} 
    \item[] Justification: Source code and reproduction scripts are publicly available at \url{https://github.com/yellowon/Jaguar}.

\item {\bf Crowdsourcing and research with human subjects}
    \item[] Question: For crowdsourcing experiments and research with human subjects, does the paper include the full text of instructions given to participants and screenshots, if applicable, as well as details about compensation (if any)? 
    \item[] Answer: \answerNA{} 
    \item[] Justification: The work does not involve crowdsourcing or human-subject experiments.

\item {\bf Institutional review board (IRB) approvals or equivalent for research with human subjects}
    \item[] Question: Does the paper describe potential risks incurred by study participants, whether such risks were disclosed to the subjects, and whether Institutional Review Board (IRB) approvals (or an equivalent approval/review based on the requirements of your country or institution) were obtained?
    \item[] Answer: \answerNA{} 
    \item[] Justification: The work does not involve human-subject research or collection of personal data.

\item {\bf Declaration of LLM usage}
    \item[] Question: Does the paper describe the usage of LLMs if it is an important, original, or non-standard component of the core methods in this research? Note that if the LLM is used only for writing, editing, or formatting purposes and does \emph{not} impact the core methodology, scientific rigor, or originality of the research, declaration is not required.
    \item[] Answer: \answerNA{} 
    \item[] Justification: LLMs are not used as an original or non-standard component of the method, implementation, or experiments.

\end{enumerate}

\end{document}